\documentclass[12pt,draftcls,onecolumn]{IEEEtran}

\IEEEoverridecommandlockouts                              

\usepackage{amsmath,color}
\usepackage{amssymb}
\usepackage{amsfonts}
\usepackage{graphicx}
\usepackage{epstopdf}
\usepackage{amsthm}
\usepackage{amsmath}
\usepackage{cancel}
\usepackage{mathrsfs}
\usepackage{mathdots}
\usepackage{euscript}
\usepackage{amscd}
\usepackage{cite}
\usepackage{placeins}
\usepackage{algorithm2e}
\usepackage{tikz}
\usetikzlibrary{snakes,arrows,shapes}

\graphicspath{{images/}}

\newtheorem{thm}{Theorem}
\newtheorem{cor}[thm]{Corollary}
\newtheorem{lemma}[thm]{Lemma}
\newtheorem{prop}[thm]{Proposition}

\newtheorem{prob}[thm]{Problem}

\newcommand{\bmat}{\left[ \begin{matrix}}
\newcommand{\emat}{\end{matrix} \right]}

\newcommand{\script}[1]{\EuScript{#1}}
\newcommand{\Tr}{\operatorname{\!{^{\mbox{\scriptsize\sf T}}}}\!}
\newcommand{\Circ}{\mathop{\rm Circ}}
\newcommand{\Span}{\mathop{\rm Span}}
\newcommand{\Frac}[2]{{\displaystyle\frac{#1}{#2}}}

\newcommand{\E}{{\mathbb E}}
\newcommand{\Rbb}{\mathbb R}

\newcommand{\Cbb}{\mathbb C}

\newcommand{\Zbb}{\mathbb Z}
\newcommand{\Tbb}{\mathbb T}
\newcommand{\xb}{\mathbf  x}
\newcommand{\yb}{\mathbf  y}
\newcommand{\zb}{\mathbf  z}
\newcommand{\wb}{\mathbf  w}
\newcommand{\vb}{\mathbf  v}
\newcommand{\eb}{\mathbf  e}

\newcommand{\gb}{\mathbf  g}
\newcommand{\hb}{\mathbf  h}
\newcommand{\ab}{\mathbf a}
\newcommand{\bb}{\mathbf  b}

\newcommand{\pb}{\mathbf  p}

 \newcommand{\cb}{\mathbf c}
\newcommand{\qb}{\mathbf q}  

\newcommand{\Ab}{\mathbf A}
\newcommand{\Bb}{\mathbf B}
\newcommand{\Cb}{\mathbf C}

\newcommand{\Ib}{\mathbf I}

\newcommand{\Mb}{\mathbf M}

\newcommand{\Pb}{\mathbf P}
\newcommand{\Qb}{\mathbf Q}
\newcommand{\Rb}{\mathbf R}
\newcommand{\Sb}{\mathbf S}
\newcommand{\Tb}{\mathbf T}

\newcommand{\Vb}{\mathbf V}

\newcommand{\Wb}{\mathbf W}

\newcommand{\gammab}{\boldsymbol{\gamma}}
\newcommand{\Sigmab}{\boldsymbol{\Sigma}}

\begin{document}
\title{A New Algorithm for  Circulant Rational Covariance Extension and Applications to Finite-interval Smoothing
\thanks{This research was supported by the China Scholarship Program (CSC).} 
}
\author{Giorgio Picci,~\IEEEmembership{Life~Fellow,~IEEE} and
Bin Zhu
\thanks{G. Picci  and Bin Zhu are with the
Department of Information Engineering, University of Padova, via
Gradenigo 6/B, 35131 Padova, Italy; e-mail:  
{\tt picci@dei.unipd.it}, {\tt niuniuneverdie@gmail.com}} 
}
\markboth{}
{Circulant Covariance Extension}
\maketitle

\begin{abstract}
The partial stochastic realization of periodic processes from finite covariance data has    recently been solved  by  Lindquist and Picci based on convex optimization of a generalized  entropy functional.  The meaning  and the role of this criterion   have an unclear origin.    In this paper  we propose a  solution based on a nonlinear generalization of the classical Yule-Walker type equations and on a new iterative algorithm which  is shown to converge to the same (unique) solution of the variational problem. This provides a conceptual link to the variational principles and at the same time yields a robust algorithm which can for example be    successfully applied to finite-interval smoothing problems  providing a simpler procedure if compared with the classical Riccati-based calculations. 
 \end{abstract}
 
 
 \section{Introduction}
 
 The rational covariance extension   is an important problem with applications in signal processing, identification and   control which has been extensively discussed in the literature; see, e.g., \cite{Kalman,Gthesis,Georgiou1,BLGM1,BGuL,SIGEST,Byrnes-L-97,BEL1,BEL2,PEthesis,Pavon-F-12,Carli-FPP} and references therein. It is usually formulated and discussed for stationary processes defined on the whole integer line.  In practice, however  data are often only available on a finite time interval.   Modeling of finite-interval processes on the other hand, can not be based on a trivial periodic extension as for deterministic signals but requires instead   a  matrix completion of Toeplitz covariance  matrices to  a {\em circulant structure}  \cite{Carli-FPP}.  Covariance extension for finite-interval  stochastic processes, then leads to  circulant matrix completion problems and to partial stochastic realizations in the form of bilateral periodic ARMA models. The problem has been solved  in \cite{LPcirculant}  for scalar processes   and  in \cite{LMPcirculant} for a certain class of vector processes admitting   a matrix fraction representation with a scalar numerator polynomial.

 As pointed out in \cite{LPcirculant} the circulant rational covariance extension theory provides a fast approximation procedure for solving the classical rational covariance extension problem on the integer line and thereby provides a reliable and  fast numerical procedure to do stationary model approximation on a finite interval, as it is based on finite Fourier transform (DFT) and use of the fast Fourier transforms (FFT). This for example leads to approximate solutions of the finite interval smoothing problem in a stationary setting which avoids Riccati equations altogether and also avoids the well-known  transient phenomena at the endpoints of the interval as it subsumes stationarity also on the finite interval. In the present paper we shall provide  evidence that this  also holds in the multivariable case.

 In \cite{Carli-FPP}  Carli, Ferrante, Pavon and Picci presented a maximum-entropy approach to the circulant covariance extension problem, thereby providing a procedure for determining the unique bilateral periodic AR model matching a given partial  covariance sequence. Generalized entropy maximization has however been the basis for  much earlier work on stationary covariance extension on the integer line. The first complete solution of this problem has actually been obtained by   recasting it  in the context of the optimization-based theory of moment problems with rational measures developed in \cite{BGuL,SIGEST,BGL1,BEL2}, \cite{BLmoments,BLkrein,Georgiou3,Georgiou-L-03}. If we except the general philosophical introduction by Dempster \cite{Dempster-72}, the significance and the actual necessity of this variational approach  has so far been an elusive question. 
 
 In this paper we shall attempt to reformulate the covariance extension problem in the more familiar setting  of Yule-Walker type covariance equations. Although the resulting equations turn out to be nonlinear, a very natural iterative solution is apparent from their structure and this iterative solution provides an unexpected link with the variational solution.

The outline of the paper is as follows. 

In Section~\ref{SecInteger} we review the covariance extension problem on the integer line formulated as maximization of a generalized entropy functional, discuss representation by unilateral ARMA processes and  propose an algorithm for solving iteratively the nonlinear Yule-Walker type equation for the denominator parameters of the model. A first statement of the interpretation of this iterative algorithm as a quasi-Newton  solution of the generalized entropy maximization is provided.

In Section~\ref{covextsec} the covariance extension problem for periodic processes is formulated first in terms of the spectral density on the discrete unit circle and then in terms of periodic ARMA models. A generalization of the nonlinear Yule-Walker  equations and an iterative algorithm for their solution are introduced. It is shown (Theorem~\ref{propos_quasiNewton})  that the algorithm is a quasi-Newton  procedure to minimize the dual of the generalized entropy functional expressed in terms of the circulant ARMA polynomials.

Section~\ref{Proof} contains  an explicit layout of the algorithm and the proof that the algorithm actually converges to a minimum. Some subtleties about discrete spectral factorization are also discussed.

A generalization of the results to the vector case is briefly presented in Section~\ref{SecVector}. Finally an application of the theory to   the finite-interval smoothing problem is presented in Section~\ref{Smoothing}. The paper also contains an appendix with some background material on circulant matrices and spectral factorization on the discrete unit circle.

\subsection*{Symbols}
Throughout this paper $\E$ denotes mathematical expectation and $\Zbb$ the set of integer numbers.  Boldface symbols are used for circulant matrices and related quantities. The symbol  $ \overline{\Span}$    stands for closed linear hull in the standard Hilbert space of second order random variables, see e.g. \cite{LPBook}. In our setting, to match Fourier-domain notations,     polynomials will be written  as functions of the  indeterminate $z^{-1}$.  In particular, the definition  of {\em Schur polynomials}, which   normally can only have zeros inside the complex unit disk, $\{|z|<1\}$ has been slightly modified  to accomodate this convention; see Sec.\ref{SecInteger}.  The set of real Schur polynomials of degree $n$ will be denoted by the symbol $\mathcal{S}_n$.

\section{The  rational covariance extension problem on  the integer line.}\label{SecInteger}
For future reference we first  state   the well-known extension  problem for a scalar covariance sequence. Let $\{y(t)\}$ be a  real zero-mean (wide-sense) stationary process on $\Zbb$  whose first $n+1$ covariance lags  $c_k:= \E\{y(t+k) y(t)\}\,;\, k=0,1,\ldots n$ are  arranged in a symmetric  Toeplitz matrix
\begin{equation}
\label{Toeplitz}
\Tb_n=\begin{bmatrix} c_0&c_1&c_2&\cdots&c_n\\
				c_1&c_0&c_1&\cdots& c_{n-1}\\
				c_2&c_1&c_0&\cdots&c_{n-2}\\
				\vdots&\vdots&\vdots&\ddots&\vdots\\
				c_n&c_{n-1}&c_{n-2}&\cdots&c_0
 		\end{bmatrix}\,, \qquad n\in \Zbb_+\,.
\end{equation}
 Following \cite{BLGM1,BGuL},  consider the following  version of the classical covariance completion problem stated in terms of rational spectral densities.

\begin{prob}\label{RCEP}
	Suppose we are given a symmetric pseudo-polynomial
	\[P(z)=\sum_{k=-n}^{n}p_kz^{-k},\quad p_{-k}=p_k,\]
	which takes positive values on the unit circle; i.e. $P(e^{i\theta})>0 \,;\, \theta \in [-\pi,\,\pi]$ and $n+1$ real numbers $c_0,c_1,\dots,c_n$ such that the Toeplitz matrix $\mathbf{T}_n$ is positive definite. We want to determine a pseudo-polynomial 
	\[Q(z)=\sum_{k=-n}^{n}q_kz^{-k},\quad q_{-k}=q_k\] such that  $\Phi(e^{i\theta}):=P(e^{i\theta})/Q(e^{i\theta})$ is a  spectral density satisfying the moment conditions
	\begin{equation}\label{moment_cond}
	\int_{-\pi}^{\pi}e^{ik\theta}\Phi(e^{i\theta})\frac{d\theta}{2\pi}=c_k,\quad k=0,1,\dots,n.
	\end{equation}
\end{prob}

There is  an alternative formulation of this problem  in terms of ARMA models. Let $a(z),b(z)$ be a pair of polynomials of degree $n$ in the indeterminate $z^{-1}$, interpreted as the unit delay operator on real sequences,
\begin{equation}\label{def_poly_ab}
a(z):=\sum_{k=0}^{n}a_kz^{-k},\quad b(z):=\sum_{k=0}^{n}b_kz^{-k}.
\end{equation}
and consider the  process  $y$ defined on the whole integer line, described by the  ARMA model 
\begin{equation}\label{ARMA}
\sum_{k=0}^{n}a_ky(t-k)=\sum_{k=0}^{n}b_kw(t-k),\quad t\in\mathbb{Z},
\end{equation}
or,   symbolically as $ a(z)y(t)=b(z)w(t) $ where $\{a_k,b_k\}$ are the  coefficients in \eqref{def_poly_ab}, and $w$ is a white noise process of variance $\mathbb{E}[w(t)]^2=\sigma^2$. To guarantee uniqueness of the representation we shall hereafter take $z^na(z)$ and $z^nb(z)$ to be Schur polynomials\footnote{Herefter a polynomial of degree $n$, $p(z)$, in the indeterminate $z^{-1}$ will be called {\em Schur} if  all  zeros of $z^np(z)$ are inside the unit circle.} and normalize  $b(z)$ to be  monic (i.e. such that $b_0=1$).  Consider then the following problem:

\begin{prob}\label{ARMA_ident_ordy}
	Suppose we are given a monic Schur polynomial $b(z)$ of degree $n$ and $n+1$ numbers $c_0,c_1,\dots,c_n$ such that the Toeplitz matrix $\mathbf{T}_n$ is positive definite. We want to determine the coefficients of the polynomial $a(z)$ such that the the first $n+1$ covariance lags of the ARMA process defined by (\ref{ARMA}) are  equal to $c_0,c_1,\dots,c_n$.
\end{prob}
 This problem has been discussed and solved by variational techniques similar  to  those used for  Problem \ref{RCEP}, by P. Enquist   \cite{PEthesis,En,PE}.
It is easily seen that any solution to Problem \ref{ARMA_ident_ordy}   solves Problem \ref{RCEP} yielding  a spectral density
\begin{equation}\label{spectFactinz}
\Phi(e^{i\theta}):=\sigma^2\left|\frac{b(e^{i\theta})}{a(e^{i\theta})}\right|^2=\sigma^2\left.\frac{b(z)b(z^{-1})}{a(z)a(z^{-1})}\right|_{z=e^{i\theta}},
\end{equation}
with the denominator $Q(z):=a(z^{-1})a(z)/\sigma^2$. Conversely, given a pair of pseudo-polynomials $P(z),\,Q(z)$ solving Problem  \ref{RCEP},  and their outer polynomial spectral factors $b(z),\,a(z)$, i.e.
\begin{equation} 
P(z)= b(z)b(z^{-1})\,,\qquad Q(z)= a(z)a(z^{-1})
\end{equation}
 any solution of Problem  \ref{RCEP} yields a solution of Problem   \ref{ARMA_ident_ordy} where the noise $w(t)$ in the ARMA model \eqref{ARMA}  has unit variance. 
Our initial goal will be to formulate the estimation of the  $\{a_k\}$ parameters  in the more familiar frame of Yule-Walker-type equations. This approach will turn out to provide an interesting link with the mainstream procedure in the literature based on  generalized entropy minimization \cite{BGuL,SIGEST}.

\subsection*{Characterization of the AR coefficients}
 
Being defined in terms of  Schur polynomials, the ARMA model \eqref{ARMA} is a causal and causally invertible innovation model. This means that (as mentioned e.g. in \cite{SS}), 
\begin{itemize}
	\item 
	The process $\{y(t)\}$ has a  representation 
	\begin{equation}\label{MA_infty}
	y(t)=\sum_{k=0}^{\infty}\gamma_k\, w(t-k)
	\end{equation}
	which is convergent in mean  square  (causality).
	\item 
	Conversely, the process $\{w(t)\}$ can be represented in terms of $\{y(t)\}$
	\[w(t)=\sum_{k=0}^{\infty}\pi_k\, y(t-k)\]
	the sum also converging in mean square  (invertibility).	
	\item 
	The innovation property: $\overline{\Span}\{y(k),k\leq t\}= \overline{\Span}\{w(k),k\leq t\} $ holds for all $t\in \Zbb$ so that in particular,
	\begin{equation}\label{causality}
	w(t)\perp \overline{\Span}\{y(k),k\leq t-1\}.
	\end{equation}
\end{itemize}

Consider the representation (\ref{MA_infty}) and let  $\gamma(z):=\sum_{k=0}^{\infty}\gamma_k\, z^{-k}$; then it is clear that
\[
\gamma(z)=\frac{b(z)}{a(z)}\implies b(z)=\gamma(z)a(z)\,,
\]
so that, by matching the coefficients of polynomials  one can see that for $j>n$,  the sequence $\{\gamma_j\}$ satisfies
a homogeneous difference equation
\[a_0\gamma_j+a_1\gamma_{j-1}+\dots+a_n\gamma_{j-n}=0\]
with the initial conditions specified by
\begin{equation}\label{IC_MAinfty_para}
\left[\begin{array}{ccccc}
a_0 & 0 & 0 & \dots & 0 \\
a_1 & a_0 & 0 & \dots & 0 \\
a_2 & a_1 & a_0 & \dots & 0 \\
\vdots & & \ddots & \ddots & \vdots \\
a_n & \dots & a_2 & a_1 & a_0 \\
\end{array}\right]
\left[\begin{array}{c}
\gamma_0 \\
\gamma_1 \\
\gamma_2 \\
\vdots \\
\gamma_n
\end{array}\right]
=\left[\begin{array}{c}
b_0 \\
b_1 \\
b_2 \\
\vdots \\
b_n
\end{array}\right],
\end{equation}
 Denoting the $n\!+\!1$-dimensional vectors  of coefficient $\{a_k,b_k\,;\, k=0,1,\ldots,n\}$ by $\ab$ and $\bb$ respectively, equation \eqref{IC_MAinfty_para} can be rewritten in vector form as
\begin{equation}
\mathbf{T}(\mathbf{a})\gammab=\mathbf{b} 
\end{equation}
where  $\gammab$ is the vector with components $\{\gamma_j,j=0,\dots,n\}$, which yields
\begin{equation}\label{gamma}
 \gammab(\mathbf{a})=\mathbf{T}(\mathbf{a})^{-1}\mathbf{b}.
\end{equation}
In general, to specify completely the model \eqref{ARMA} we need to specify  the noise variance parameter $\sigma^2$. For an innovation model, it may seem natural to  fix  $w(t)$ equal to the one-step prediction error $e(t)= y(t)-\hat{y}(t\mid t-1)$ in which case however  the impulse response $\{\gamma_k\}$ should be  normalized so  that $\gamma_0=1$. Since by convention we choose monic numerator polynomials,
($b_0=1$)  \eqref{IC_MAinfty_para} implies that  $\gamma_0 =1/a_0$ and  hence    one should constrain also $a_0$ to be  normalized to have modulus one (more precisely $1/a_0^2=1$). Adding  this extra constraint   would however be quite  inconvenient since the vector $\ab$ is determined either   as a result of an optimization problem or   updated by an iterative algorithm. A natural solution is to let  the noise variance  vary as a function of 	$\ab$ without imposing any normalization to $a_0$. In fact one can normalize to obtain $a_0=1$   once the vector $\ab$ is determined and makes up  the  coefficients of a Schur polynomial  in which case the innovation variance would simply be $ \sigma_{e}^2= \sigma^2/a_0^2$, where $\sigma^2$ is the variance parameter without  normalization which can be computed by the formula \eqref{sigma_sqr} shown below.

One can similarly derive an analogous difference equation for the covariance lags of the process $\{y(t)\}$ 
\[a_0c_j+a_1c_{j-1}+\dots+a_nc_{j-n}=0\]
valid for $j>n$, with different initial conditions which can be written in matrix form as 
\begin{equation}\label{charac_a}
\mathbf{T}_n\mathbf{a}=\sigma^2\mathbf{H}_b\gammab,
\end{equation}
where
\[\mathbf{H}_b:=\left[\begin{array}{ccccc}
	b_0 & b_1 & b_2 & \dots & b_n \\
	b_1 & b_2 & \dots & b_n & 0 \\
	b_2 & \dots & b_n & 0 & 0 \\
	\vdots & \iddots & \iddots &  & \vdots \\
	b_n & 0 & \dots & 0 & 0 \\
	\end{array}\right].\]
Combining \eqref{gamma} and \eqref{charac_a} to eliminate $\gammab$ 	one obtains a nonlinear equation for the parameter $\ab$ as a function of the covariance lags (in $\mathbf{T}_n$) and the known numerator parameters $\bb$,
\begin{equation}\label{equat_a}
\mathbf{a}=\sigma^2\mathbf{T}_n^{-1}\mathbf{H}_b\mathbf{T}(\mathbf{a})^{-1}\mathbf{b}.
\end{equation}
In this equation the variance parameter $\sigma^2$ can also be expressed in function of the parameter $\ab$ as	
\begin{equation}\label{sigma_sqr}
\sigma^2(\mathbf{a})=\frac{a_0}{b_n}\sum_{k=0}^{n}a_kc_{n-k},
\end{equation}
as obtained  from the last equation in (\ref{charac_a}).

\subsection*{An iterative  solution  of the rational Covariance Extension problem}
Since there is no explicit solution of equation \eqref{equat_a} in view,   we propose to solve it by an iterative scheme of the form
\begin{equation}\label{fixed-pt_iter_ordy}
\mathbf{a}^{(k+1)}=\sigma^2\mathbf{T}_n^{-1}\mathbf{H}_b\mathbf{T}(\mathbf{a}^{(k)})^{-1}\mathbf{b},\qquad k=0,1,\ldots
\end{equation}
which can for example be initialized with the coefficients of the maximum entropy solution    $\mathbf{a}^{(0)}=\mathbf{a}_{\mathrm{ME}}$ computed by the Levinson algorithm with data  the Toeplitz matrix $\mathbf{T}_n$. The scaling factor $\sigma^2$ is updated in each iteration by substituting  the iterate $\mathbf{a}^{(k)}$ in \eqref{sigma_sqr}.

Proving convergence of this iteration is a non-trivial matter. However Anders Lindquist has suggested us a slick proof  which makes contact with  the generalized entropy minimization approach to  Problem \ref{ARMA_ident_ordy}. The  following theorem  which we state here without proof is based on  his idea. The proof here is omitted since it will turn out to be  a particularization   of the  proof of Theorem \ref{periodThm}   given in  Sec. \ref{Proof} for the periodic Problem~\ref{ARMA_identify}.
\begin{thm}\label{propos_quasiNewton_ordy}
	The  iteration scheme $(\ref{fixed-pt_iter_ordy})$ can be interpreted as a quasi-Netwon step with a scaling parameter $\sigma^2$ for the minimization of the function
	\begin{equation}\label{func_a_ordy}
	\mathbb{J}_P(\mathbf{a})=\mathbf{a}^\top\mathbf{T}_n\mathbf{a}-\int_{-\pi}^{\pi}b(e^{-i\theta})b(e^{i\theta})\log [a(e^{-i\theta})a(e^{i\theta})]\frac{d\theta}{2\pi}.
	\end{equation}
	which has a unique minimum in the set $\mathcal{S}_n$ of real   Schur polynomials of degree $n$ with $a_0>1$.
	 \end{thm}
For, it is shown in    \cite{En},  that the optimization problem
\begin{equation}\label{opt_a_ordy}
\min_{a(z)\in\mathcal{S}_n} \;  \mathbb{J}_P(\mathbf{a}) 
\end{equation}
has a unique solution in the feasible set.\\
  Algorithm \eqref{fixed-pt_iter_ordy} with some adaptations has been successfully tested in several examples. However our main interest is in periodic processes   which we shall turn to in the next section.

\section{The covariance extension problem for periodic processes}\label{covextsec}
Consider a zero-mean  second order real  process $\{y(t)\}$,  defined on a finite interval $[-N+1,\,N]$ of the integer line $\Zbb$ and extended to all of $\Zbb$ as a periodic  process with  period $2N$ so that 
\begin{equation}
\label{periodic2N}
y(t +2kN) =y(t) 
\end{equation}
almost surely. We shall as usual, say that $\{y(t)\}$ is {\em stationary}  if  the covariance lags  $c_k:= \E\{y(t+k) y(t)\}$ do not depend on time and hence the covariance  matrix with entries $\{c_k\}$ has a Toeplitz structure. In fact, as shown in \cite{Carli-FPP} in order for  the random vector 
\begin{equation}\label{vecy}
\yb:= \bmat y(t-N+1)& y(t-N+2)&\ldots&y(t+N)\emat^{\top}
\end{equation}
to represent the restriction to $[-N+1,\,N]$ of  a periodic process on $\Zbb$,   the covariance   $\Tb_{2N}:= \E\,\yb\yb^{\top}$,   must be  a {\em circulant matrix}, namely it must have the form
\begin{align} \label{ToeplitzCirc}
\Tb_{2N}& =\begin{bmatrix} c_0&c_1&\cdots &c_N &c_{N-1}&\cdots &c_1\\
				c_1&c_0&\cdots& c_{N-1}&c_{N}& \cdots& c_2\\
				\vdots&\vdots&\ddots&\vdots&&&\vdots\\
				c_N& c_{N-1}&c_{N-2}& & &\cdots&c_{N-1}  \\
				\vdots&\vdots&\vdots&\ddots&&&\vdots\\
				c_1&\cdots &c_{N}&c_{N-1} &\cdots &c_1&c_0
 		\end{bmatrix}\,,\\
		&=\Circ\{ c_0,c_1,c_2,\dots, c_N,c_{N-1},\dots,c_2 ,c_1 \}
\end{align}
where the columns are shifted cyclically, the last component moved to the top \cite{Davis-79}.  Circulant matrices will play a key role in the following. 

By stationarity $y$ has a spectral representation
\begin{equation}\label{specy }
y(t)=\int_{-\pi}^\pi e^{it\theta}d\hat{y}(\theta), \quad \text{ where} \qquad  \E \{|d\hat{y}|^2\}=dF(e^{i\theta}) \,, 
\end{equation}
is the spectral distribution (see, e.g., \cite[p. 74]{LPBook}), and therefore
\begin{equation}
\label{F2c}
c_k:= \E\{y(t+k)y(t)\} = \int_{-\pi}^\pi e^{ik\theta}dF(e^{i\theta}). 
\end{equation}
Because of the periodicity  condition \eqref{periodic2N}, the support of the spectral distribution $dF$ must be contained in    the {\em discrete unit circle} $\mathbb{T}_{2N}:=\{\zeta_{-N+1},\zeta_{-N+2},\dots,\zeta_N\}$, where
\begin{equation}
\label{zetadefn}
\zeta_k=e^{ik\pi/N}. 
\end{equation}
As explained in the appendix, see equation (\ref{Phi2c}), one can represent $dF$ as
$dF= \Phi \, d\nu$ where $d\nu$ is a uniform discrete measure supported on $\mathbb{T}_{2N}$ and   $\Phi$ is  the discrete Fourier transform (DFT) of the sequence $(c_{-N+1},\dots,c_N)$, called the {\em spectral density} of $\yb$,    
\begin{equation}
\label{Phi}
\Phi(\zeta)=\sum_{k=-N+1}^N c_k\,\zeta^{-k}
\end{equation} 
which is in fact {\em the symbol} of the circulant matrix $\Tb_{2N}$. This is a nonnegative  function of the discrete variable $\zeta \in\mathbb{T}_{2N}$ which is strictly positive  if and only if the $2N\times 2N$ covariance matrix $\Tb_{2N}$ is positive definite, see \cite[Proposition 2]{CarliGeorgiou}, that is to say, the process is {\em full rank} which we shall assume all through this paper.

Suppose now that we are given a partial covariance sequence $c_0,c_1,\dots,c_n$ with $n< N$, such that the Toeplitz matrix $\Tb_n$ is positive definite. Consider the problem of finding an extension $c_{n+1},c_{n+2},\dots,c_{N}$ which, once imposing the periodic midpoint reflection condition $c_{N+k}=c_{N-k}$ valid for real covariances,  makes the  sequence  $c_0,c_1,\dots,c_N$ a covariance sequence of a stationary process of period $2N$. 

In general this problem, whenever feasible,  will have infinitely many solutions. We are  however interested in finite complexity solutions only  and so  we shall restrict our attention to spectral functions \eqref{Phi} which are {\em rational} in the sense that 
\begin{equation} \label{Phi=P/Q}
\Phi(\zeta)=\frac{P(\zeta)}{Q(\zeta)},
\end{equation}
where $P$ and $Q$ are symmetric pseudo-polynomials of degree (at most) $n$, that is of the form
\begin{equation}
\label{P}
P(\zeta)=\sum_{k=-n}^n p_k \zeta^{-k}, \quad p_{-k}=p_k. 
\end{equation}

A convex optimization approach to determine rational solutions $\Phi(\zeta)$ is proposed by Lindquist and Picci in \cite{LPcirculant} where a complete parametrization of all such solutions is described. Feasibility of the optimization problem can be described in the language of moment problems  \cite{KreinNudelman,BLkrein} as follows.\\
  Let $\mathfrak{P}_+(N)$ be the cone of all symmetric pseudo-polynomials \eqref{P} of degree $n$ that are positive on the discrete unit circle $\mathbb{T}_{2N}$, and let $\mathfrak{P}_+\subset\mathfrak{P}_+(N)$ be the subset of pseudo-polynomials \eqref{P} such that $P(e^{i\theta})>0$ for all $\theta\in [-\pi,\pi]$. Moreover let $\mathfrak{C}_+(N)$ be the dual cone of all 
partial covariance sequences $\cb=(c_0,c_1,\dots,c_n)$ such that 
\begin{displaymath}
\langle \cb,\pb\rangle :=\sum_{k=-n}^n c_k p_k >0 \quad \text{for all $P\in\overline{\mathfrak{P}_+(N)}\setminus\{0\}$},
\end{displaymath}
and let $\mathfrak{C}_+$ be defined in the same way as the dual cone of $\mathfrak{P}_+$. It can be shown \cite{KreinNudelman} that $\cb\in\mathfrak{C}_+$ is equivalent to the Toeplitz condition $\Tb_n>0$. Since $\mathfrak{P}_+\subset\mathfrak{P}_+(N)$, we have $\mathfrak{C}_+(N)\subset \mathfrak{C}_+$, so in general $\cb\in\mathfrak{C}_+(N)$ is a stricter condition than $\Tb_n>0$. 

Conditions for the existence and uniqueness of the solution to the extension problem can then be stated in the following way. 

\begin{thm}\label{mainthm}
Let $\cb\in\mathfrak{C}_+(N)$. Then, for each $P\in\mathfrak{P}_+(N)$ of degree $n$, there is a unique $Q\in\mathfrak{P}_+(N)$ such that the rational function
$
\Phi=\frac{P}{Q}
$
satisfies the moment conditions 
\begin{equation}
\label{momentconditions}
\int_{-\pi}^\pi e^{ik\theta}\Phi(e^{i\theta})d\nu(\theta) =c_k, \quad k=0,1,\dots,n.
\end{equation}
\end{thm}
Consequently the family of solutions \eqref{Phi=P/Q} of the periodic covariance extension problem stated above is parameterized by $P\in\mathfrak{P}_+(N)$ in a bijective fashion. A key result of the theory is that, for any $P\in\mathfrak{P}_+(N)$, the corresponding unique $Q\in\mathfrak{P}_+(N)$ can be obtained by convex optimization.  

\begin{thm}\label{optthm}
Let $\cb\in\mathfrak{C}_+(N)$ and $P\in\mathfrak{P}_+(N)$. Then  the problem to maximize
\begin{equation}
\label{primal}
\mathbb{I}_P(\Phi) =\int_{-\pi}^\pi  P(e^{i\theta})\log \Phi(e^{i\theta})d\nu
\end{equation}
subject to the moment conditions \eqref{momentconditions} has a unique solution  of the form \eqref{Phi=P/Q}, where $Q$ is the unique optimal solution of the problem to minimize over all $Q\in\mathfrak{P}_+(N)$ the dual functional
\begin{equation}\label{dual}
\mathbb{J}_P(Q)= \langle \cb,\qb\rangle -\int_{-\pi}^\pi  P(e^{i\theta})\log Q(e^{i\theta})d\nu
\end{equation}
where $\qb:=(q_0,q_1,\dots,q_n)$ are the coefficients of $Q$.  The functional $\mathbb{J}_P$ is strictly convex. 
\end{thm}

We refer the reader to \cite{LPcirculant} for the proofs. Theorems \ref{mainthm} and \ref{optthm} are discrete versions of  results in \cite{BGuL,SIGEST} valid for the integer line $\Zbb$.
The solution corresponding to $P=1$ is called the {\em maximum-entropy solution\/} by virtue of \eqref{primal}.

\subsection*{Covariance extension by unilateral periodic ARMA models}\label{CovARMAsect}
Periodic processes can be conveniently seen as being defined on the finite group $\Zbb_{2N}$ made of the discrete interval $[-N+1,N]$ with arithmetics modulo $2N$.   There $\{y(t)\}$ can be  represented as a $2N$-dimensional vector as in \eqref{vecy}.  We are interested in    periodic processes which can   be represented by  {\em unilateral} ARMA models of the form
\begin{equation}\label{ARMA_periodic}
\sum_{k=0}^{n}a_ky(t-k)=\sum_{k=0}^{n}b_kw(t-k),\qquad t\in\mathbb{Z}_{2N}
\end{equation}
where $\{w(t)\}$ is a periodic  white noise process on  $\Zbb_{2N}$ of variance $\sigma^2$ and $\{a_k\}$ and $\{b_k\}$ are the coefficients of two  Schur   polynomials $a(z),\,b(z)$ where we shall  again take    $b_0=1$. An alternative is to leave $b_0$ free and normalize   $\sigma^2=1$, depending on convenience.  To impose periodicity to \eqref{ARMA_periodic} we need to impose  periodic boundary conditions at the endpoints; i.e.
\begin{equation}\label{BC_prd}
y(-N)=y(N),\dots,\ y(-N-n+1)=y(N-n+1)\, 
\end{equation}
 which leads to a circulant matrix representation of the model \eqref{ARMA_periodic}. Introducing the  vector notation
\[\mathbf{w}=[w(-N+1),w(-N+2),\dots,w(0),w(1),\dots,w(N)]^\top,\]
 we   have $\mathbb{E}\{\mathbf{w}\mathbf{w}^\top\}=\sigma^2\Ib_{2N}$ and a  unilateral ARMA model describing a scalar-valued periodic stationary process $\yb$ may then be rewritten  compactly as a matrix-vector equation
\begin{equation}\label{ARMA_mat}
\mathbf{A}\mathbf{y}=\mathbf{B}\mathbf{w},
\end{equation}
in which $\mathbf{A}$ and $\mathbf{B}$ are $2N\times2N$ nonsingular lower-triangular circulant matrices of bandwidth $n$
\[\begin{split}
\mathbf{A} & =\mathrm{Circ}\{a_0,\,a_1,\dots,a_n,0,\dots,0\}, \\
\mathbf{B} & =\mathrm{Circ}\{1,\,b_1,\dots,b_n,0,\dots,0\}.
\end{split}\]
That a large class of  periodic processes can admit unilateral ARMA representations is shown in \cite{Pi-MTNS016}\footnote{In this conference paper an important positivity condition is overlooked; the correct representability condition is discussed later   in this paper.}, and is also surveyed in Appendix \ref{unilateralARMAsec}. We consider now the analog of Problem   \ref{ARMA_ident_ordy} for periodic processes.

\begin{prob}[The Circulant Rational Covariance extension Problem (CRCEP)]\label{ARMA_identify}
	Suppose we are given the $n$ coefficients $\{b_k\,;\,k=1,2,\ldots,n\}$ of a Schur   polynomial and $n+1$ real numbers $c_0,c_1,\dots,c_n$ such that the Toeplitz matrix \eqref{Toeplitz} is positive definite. We want to determine the coefficients $\{a_k\}$    such that the first $n+1$ covariance lags of the periodic process $\{y(t)\}$ in \eqref{ARMA_periodic} match the sequence $\{c_k\,;\, k=0,1,\ldots,n\,\}$.
\end{prob}
If we were able to solve this problem, i.e., to obtain the matrix $\mathbf{A}$, letting $\mathbf{\Sigma}:=\mathbb{E}\{\mathbf{y}\mathbf{y}^\top\}$ and taking covariance on both sides of the equation (\ref{ARMA_mat}), we would have
\[\mathbf{A}\mathbf{\Sigma}\mathbf{A}^\top=\sigma^2 \,\mathbf{B}\mathbf{B}^\top\,.\]
Taking inverses and rearranging terms, due to the commutativity of circulant matrices, this would lead to a representation of $\Sigmab$ as the ratio of two positive circulants, i.e.
\[\begin{split}
\mathbf{\Sigma}  =\sigma^2 \mathbf{A}^{-1}\mathbf{B}\mathbf{B}^\top\mathbf{A}^{-\top}  
  & = \sigma^2(\mathbf{A}\mathbf{A}^\top)^{-1}\mathbf{B}\mathbf{B}^\top \\
  &: =\sigma^2 \mathbf{Q}^{-1}\mathbf{P},
\end{split}\]
where $\mathbf{Q}:=\mathbf{A}\mathbf{A}^\top$, $\mathbf{P}:=\mathbf{B}\mathbf{B}^\top$. This procedure would then solve the same circulant rational covariance extension problem    discussed in \cite{LPcirculant,LP}. Here $\mathbf{B}$ is  a circulant factor of the bilaterally $n$-banded circulant matrix $\mathbf{P}$ having the given symbol $P(\zeta)$, see Appendix \ref{unilateralARMAsec}. 

We shall show in the next sections that Problem \ref{ARMA_identify} can be converted   to the solution of a nonlinear equation similar to \eqref{equat_a}. The solution can be obtained by an iterative algorithm which, although looking similar to \eqref{fixed-pt_iter_ordy} turns out to be more difficult to analyze. The proof of its convergence will be one of the main results of this paper.

\subsection*{Spectral representation of periodic ARMA models}

In terms of the discrete Fourier transform (DFT) of the random variables $\{y(t),\,w(t)\,; t\in\mathbb{Z}_{2N}\}$ defined as
\begin{equation}
\hat{y}(\zeta_k)=\!\!\!\!\!\sum_{t=-N+1}^{N}y(t)\zeta_k^{-t},\quad \hat{w}(\zeta_k)=\!\!\!\!\!\!\sum_{t=-N+1}^{N}w(t)\zeta_k^{-t},\quad k\in \mathbb{Z}_{2N}
\end{equation}
the model  (\ref{ARMA_periodic}) can be rewritten
\begin{equation}\label{ARMA_freq}
a(\zeta)\hat{y}(\zeta)=b(\zeta)\hat{w}(\zeta),\qquad \zeta\in\mathbb{T}_{2N}
\end{equation}
where the polynomials $a(\zeta),\; b(\zeta)$ are defined in 	terms of the coefficients of the model \eqref{ARMA_periodic}as
\begin{equation}\label{poly_ab-DFT}
a(\zeta):=\sum_{k=0}^{n}a_k\, \zeta^{-k},\quad  b(\zeta):=\sum_{k=0}^{n}b_k \, \zeta^{-k}\,.
\end{equation}
Here we shall choose to fix $b_0=1$.\\
The solution of \eqref{ARMA_freq} can formally be written as
\begin{equation}\label{ARMA_DFT}
\hat{y}(\zeta)=\frac{b(\zeta)}{a(\zeta)}\hat{w}(\zeta),\qquad \zeta\in\mathbb{T}_{2N},
\end{equation}
whence, denoting the inverse DFT of $\frac{b(\zeta)}{a(\zeta)}$ by $\gamma:= \{\gamma_k\,;\, k= -N+1,\ldots,N\}$,   one obtains  a one-sided   representation of $\{y(t)\}$ in terms of the input noise $\{w(t),t\in\mathbb{Z}_{2N}\}$
\begin{equation}\label{two-sided}
y(t)=\sum_{s=-N+1}^{N}\gamma_{t-s}w(s).
\end{equation}
Since the DFT $\hat{w}(\zeta_k)$ satisfies 
\[\Frac{1}{2N}\,\mathbb{E}\left[\hat{w}(\zeta_k)\overline{\hat{w}(\zeta_l)}\right]=\sigma^2\delta_{kl} \]
the spectral density of $\{y(t)\}$ is   readily obtained as
\begin{equation}\label{spec_disc}
\Phi(\zeta_k)=\Frac{1}{2N}\, \mathbb{E}\left[\hat{y}(\zeta_k)\overline{\hat{y}(\zeta_k)}\right]=\sigma^2\frac{b(\zeta_k)b(\zeta_k^{-1})}{a(\zeta_k)a(\zeta_k^{-1})}.
\end{equation}
Notice that the sequence $\{\gamma_t\}$ is naturally periodic and is readily computable from the polynomials $a(\zeta)$ and $b(\zeta)$ via inverse FFT. We note for future use that 
$\mathbb{E}[y(t)w(s)]=\sigma^2\gamma_{t-s}.$
One can see that the one-sided representation (\ref{two-sided}) in the periodic case is different from (\ref{MA_infty}), since here we do not have  causality and $y(t)$  depends on $\{w(s)\}$ over the whole interval $[-N+1,N]$.

One can also rewrite (\ref{two-sided}) in  matrix notation as
\begin{equation}\label{circ_one_sided}
\mathbf{y}=\mathbf{\Gamma}\mathbf{w},
\end{equation}
where $\mathbf{\Gamma}\!\!\!=\!\! \mathrm{Circ}\{\gamma_0,\gamma_1,\dots,\gamma_N,\gamma_{-N+1},\dots,\gamma_{-1}\}\in\mathbb{R}^{2N\times 2N}$, the relative  symbol being 
\begin{equation}\label{GammaSymbol}
\Gamma(\zeta):=\sum_{t=-N+1}^{N}\gamma_t\,\zeta^{-t}=\frac{b(\zeta)}{a(\zeta )}.
\end{equation}
Note that the matrix $\mathbf{\Gamma}$ has a very simple expression  in terms of the  circulant matrices of coefficients $\Ab,\,\Bb$, since from  spectral theory (see Appendix \ref{AppA})  we have
\begin{equation}\label{GammaAB}
\mathbf{\Gamma}=\mathbf{A}^{-1}\mathbf{B}.
\end{equation}
Now, taking covariances on both sides of (\ref{circ_one_sided}), we obtain the circulant factorization
$ \mathbf{\Sigma}=\sigma^2\mathbf{\Gamma}\mathbf{\Gamma}^\top $
from which, combining the model equation (\ref{ARMA_mat}) with (\ref{circ_one_sided}), it is easily   seen that
\begin{equation}\label{circ_relation}
\mathbf{A}\mathbf{\Sigma}=\sigma^2\mathbf{B}\mathbf{\Gamma}^\top\,.
\end{equation}
From this relation we can now proceed to derive the nonlinear equation and an iterative scheme for the coefficient vector of the polynomial $a(\zeta)$.

Multiply    (\ref{ARMA_periodic}) on both sides by
$ y(t-j)=\sum_{s=-N+1}^{N}\gamma_{t-j-s}w(s),$
 and take expectation to obtain
\begin{equation} \label{ARMAforc}
\sum_{k=0}^{n}a_kc_{k-j}=\sigma^2\sum_{k=0}^{n}b_k\gamma_{k-j},\qquad j=0,1,\dots,n.
\end{equation}
This is a system of equations which can be written in   matrix form as
\begin{equation}\label{AR_coeff}
\mathbf{T}_n\mathbf{a}=\sigma^2\mathbf{T}_\gamma\mathbf{b}
\end{equation}
where now 
$$
\mathbf{T}_\gamma = \bmat
\gamma_0 & \gamma_1 & \dots & \gamma_n \\
\gamma_{-1} & \gamma_0 &   & \vdots \\
\vdots &   & \ddots & \gamma_1 \\
\gamma_{-n} & \dots & \gamma_{-1} & \gamma_0 \\
\emat
$$ is a full Toeplitz matrix. Since  for fixed $\bb$ the sequence $\gamma$ is a function of $\ab$ we shall denote $\mathbf{T}_\gamma$ as $\mathbf{T}_\gamma(\ab)$. It can   be computed say by inverse FFT of \eqref{GammaSymbol} or by \eqref{GammaAB}. In fact, (\ref{AR_coeff}) is just the $(n+1)\times(n+1)$ upper-left corner of the matrix equation (\ref{circ_relation}). The resulting nonlinear equation for $\ab$ is a bit more implicit than  \eqref{equat_a}. However a similar iterative scheme can be devised to solve it, say 
\begin{equation}\label{fixed-point_iteration}
\mathbf{a}^{(k+1)}=\sigma^2(\ab^{(k)})\,\mathbf{T}_n^{-1}\mathbf{T}_\gamma(\mathbf{a}^{(k)})\mathbf{b},\qquad k=0,1,\ldots
\end{equation}
with an initialization in the set $\mathcal{S}_n$ of Schur polynomials of degree $n$,
e.g., the coefficients  of the Levinson polynomial for the ordinary Toeplitz covariance extension of $\Tb_n$. The key observation here is that this algorithm is actually a numerical implementation of the variational solution of the periodic moment problem (Theorem \ref{optthm}).
\begin{thm}\label{propos_quasiNewton}
	The  iteration $(\ref{fixed-point_iteration})$ where the scaling parameter  $\sigma^2(\ab^{(k)})$ is updated in each iteration by the rule
\begin{equation}\label{sigmasq}
\sigma^2(\mathbf{a})=\left.\sum_{k=0}^{n}c_ka_k \middle/ \sum_{k=0}^{n}\gamma_kb_k\right.,\qquad (b_0=1)\,,
\end{equation}
can be interpreted as a quasi-Netwon step  for the minimization of the function
	\begin{equation}\label{func_a}
	\mathbb{J}_P(\mathbf{a})=\mathbf{a}^\top\mathbf{T}_n\mathbf{a}-\int_{-\pi}^{\pi}b(e^{i\theta})b(e^{-i\theta})\log [a(e^{i\theta})a(e^{-i\theta})]\,d\nu\,.
	\end{equation}
\end{thm}
Since \eqref{sigmasq} is exactly     equation \eqref{ARMAforc} written for $j=0$ and since fixing $b_0=1$ in this equation  determines $\sigma^2$ uniquely, we can get an immediate partial result about the convergence of  the iterates $\sigma^2(\ab^{(k)})\,; k=0,1,\ldots$. 
	\begin{prop}	Let $\bb$ be fixed; if $\ab^{(k)}$ converges to the parameters $\ab$ of an ARMA model solution of Problem~\ref{ARMA_identify}, then $\sigma^2(\ab^{(k)})$ will converge to the  variance of the white noise $\{w(t)\}$ relative to  the same model.  
\end{prop}
There is an equivalent version of the iteration \eqref{fixed-point_iteration}    where $b_0$ is not normalized and $\sigma^2$ is fixed equal to $1$. This version does  look  like a quasi-Newton method with a fixed stepsize which we do not recommend.  It is well known that a quasi-Newton method with a fixed stepsize, see e.g. (\ref{quasi_Newton}) below, may keep on chattering between two or more values without converging. This behavior was  occasionally encountered in simulations.
\subsection*{Proof of Theorem \ref{propos_quasiNewton}}
The proof follows the lemma stated below.
\begin{lemma}\label{propos_grdt}
	The gradient of $\mathbb{J}_P(\mathbf{a})$ satisfies
	\begin{equation}\label{Gradient1}
	\frac{1}{2}\,\nabla\mathbb{J}_P(\mathbf{a})= \mathbf{T}_n\mathbf{a}-\mathbf{T}_\gamma(\mathbf{a})\mathbf{b} =[\mathbf{T}_n-\mathbf{T}_n(\mathbf{a})] \, \mathbf{a}.
	\end{equation}
	where $\mathbf{T}_n(\mathbf{a})$ is the $(n+1)\times(n+1)$ upper Toeplitz submatrix of the covariance matrix corresponding to the discrete spectral density $\Phi(\zeta):=P(\zeta)/|a(\zeta)|^2$.
\end{lemma}
\begin{proof}
   A direct computation of the gradient  yields 
   \begin{equation}
   \nabla\mathbb{J}_P(\mathbf{a})=2\mathbf{T}_n\mathbf{a}-\int_{-\pi}^{\pi}b(e^{i\theta})b(e^{-i\theta})\left\lbrace\frac{1}{a(e^{i\theta})}\left[\begin{array}{c}
   1 \\
   e^{-i\theta} \\
   \vdots \\
   e^{-in\theta} \\
   \end{array}\right]+\frac{1}{a(e^{-i\theta})}
   \left[\begin{array}{c}
   1 \\
   e^{i\theta} \\
   \vdots \\
   e^{in\theta} \\
   \end{array}\right]\right\rbrace d\nu,
   \end{equation}
in which the left term of the sum inside the brace can be written as
	\[\frac{b(e^{i\theta})}{a(e^{i\theta})}
	\left[\begin{array}{c}
	1 \\
	e^{-i\theta} \\
	\vdots \\
	e^{-in\theta} \\
	\end{array}\right]
	\left[\begin{array}{cccc}
	1 &	e^{ i\theta} & \dots & e^{ in\theta} \\
	\end{array}\right]
	\left[\begin{array}{c}
	b_0 \\
	b_1 \\
	\vdots \\
	b_n \\
	\end{array}\right],\]
	so that  this part of the integral becomes the sum 
	\[\frac{1}{2N}\sum_{j=-N+1}^{N}\frac{b(\zeta_j)}{a(\zeta_j)}
	\left[\begin{array}{cccc}
	1 &	\zeta_j & \dots & \zeta_j^{n}  \\
	\zeta_j^{-1} & 1 &  & \vdots \\
	\vdots &   & \ddots & \zeta_j  \\
	\zeta_j^{-n} & \dots	& \zeta_j^{-1} & 1 \\
	\end{array}\right]\mathbf{b}
	=\mathbf{T}_\gamma(\mathbf{a})\mathbf{b}\]
	since, as seen in Sect. \ref{CovARMAsect}, 
		\[\frac{1}{2N}\sum_{j=-N+1}^{N}\frac{b(\zeta_j)}{a(\zeta_j)}\zeta_j^{-k}=\gamma_{-k},\quad k=-n,\dots,n.\]
	Computation involving the other term in the integral is similar, yielding  in fact the same result, so that
	\begin{equation}\label{eqn_grdt}
	\nabla\mathbb{J}_P(\mathbf{a})=2[\mathbf{T}_n\mathbf{a}-\mathbf{T}_\gamma(\mathbf{a})\mathbf{b}].
	\end{equation}
	Recall now that the entries of the vector $\mathbf{T}_\gamma(\mathbf{a})\mathbf{b}$ are  the initial  segment of length $n+1$ of the convolution string $\sum_{k=0}^{n}b_k\gamma_{-i+k}  \,;\, i\in \Zbb_{2N}$ whose DFT is the product
\[ b(\zeta) \Frac{b(\zeta^{-1})}{a(\zeta^{-1}) }	= \Frac{b(\zeta) b(\zeta^{-1})}{a(\zeta)a(\zeta^{-1}) } a(\zeta)= \Phi(\zeta) a(\zeta)
\]
which has inverse DFT the first column of $\Sigmab \times \Circ \{a_0,\,a_1,\ldots,a_n,\,0,\ldots,0\}$, $\Sigmab$ being the covariance corresponding to the spectral density $\Phi(\zeta) $. It follows that  in   matrix notation
	\[\mathbf{T}_\gamma(\mathbf{a})\mathbf{b}=\mathbf{T}_n(\mathbf{a})\mathbf{a}\]
which proves \eqref{Gradient1}.	
\end{proof}

If  for simpilicity we normalize to $\sigma^2=1$, then the iteration (\ref{fixed-point_iteration}) can be  written as
	\[\mathbf{T}_n[\mathbf{a}^{(k+1)}-\mathbf{a}^{(k)}]=\mathbf{T}_\gamma(\mathbf{a}^{(k)})\mathbf{b}-\mathbf{T}_n\mathbf{a}^{(k)}=-\frac{1}{2}\nabla\mathbb{J}_P(\mathbf{a}^{(k)})\]
which is  the  quasi-Newton step
	\begin{equation}\label{quasi_Newton}
	\mathbf{a}^{(k+1)}=\mathbf{a}^{(k)}-\frac{1}{2}\mathbf{T}_n^{-1}\nabla\mathbb{J}_P(\mathbf{a}^{(k)}).
	\end{equation}
Introducing instead  the scaling parameter $\sigma^2$,    the recursion  looks   like
	\begin{equation}\label{quasi_Newton_scale}
	\mathbf{a}^{(k+1)}=\sigma^2(\ab^{(k)})\,\left[\mathbf{a}^{(k)}-\frac{1}{2}\mathbf{T}_n^{-1}\nabla\mathbb{J}_P(\mathbf{a}^{(k)})\right].
	\end{equation}

\section{Proof of convergence}\label{Proof}

Before proving convergence of the algorithm, we shall   need to clarify  how the optimization of the functional \eqref{func_a} relates to the solution of  Problem \ref{ARMA_identify}.

The convex optimization approach to determine the denominator $Q(\zeta)$ of a solution to the periodic covariance extension problem   was  reviewed in Sect. \ref{covextsec}.  Theorem \ref{mainthm} gives   the main result. We   want to derive  an equivalent statement regarding the existence and uniqueness of optimal spectral factors $a(\zeta)$ of $Q(\zeta)$. There is a difficulty here since, as remarked in the appendix, $Q(\zeta)$ admits spectral factors if and only if its extension $Q(z)$ to the unit circle does so, but Theorem \ref{mainthm} only states that the optimal  $\hat Q(\zeta)$ is positive on the discrete set $\Tbb_{2N}$; said in other words, only belongs to $\mathfrak{P}_+(N)$ but not necessarily to $\mathfrak{P}_+$.
However (Lemma \ref{positivity}) if $N$ is  large, the extension $Q(z)$ will be positive. On the other hand, the key sufficient condition $\cb\in\mathfrak{C}_+(N)$ of Theorem \ref{mainthm} also requires $N$ to be larger than some $N_0$  \cite[Proposition 6]{LPcirculant}.  

We shall henceforth say that Problem \ref{ARMA_identify} is {\em feasible} if $N\geq N_0$ is large enough to guarantee that $\hat{Q}(\zeta)$ has a positive extension to the unit circle.

For a feasible problem a polynomial spectral factorization   $\hat{Q}(\zeta)=\hat{a}(\zeta)\hat{a}(\zeta^{-1})$ exists and  can  be computed  exactly  as   in the $z$-domain. See   Section \ref{sec_factorizability} of the appendix for more details. Clearly,  feasibility depends on the data, in particular on  the polynomial  $P(\zeta)$ which should be chosen so as to admit  an extension which is also positive on the whole circle; i.e. lies in $\mathfrak{P}_+$ and hence also admits  spectral factors.

All solutions of a feasible  Problem \ref{ARMA_identify}   can then be identified with nonsingular and   invertible spectral  factors of the solution \eqref{Phi=P/Q}.

Similar to what is done in \cite{En} for the ordinary covariance extension on $\Zbb$, we have reparameterized the functional  (\ref{dual}) in terms of the outer spectral factors $a(\zeta)$. Consider to this end the map 
\[T:\mathcal{S}_n\to\mathfrak{P}_+,\quad a(\zeta)\mapsto T(a):=a(\zeta)a(\zeta^{-1}).\]
By the same argument as in \cite[Section 3]{En}, one can show that the map $T$ is bijective, continuously differentiable and has a nonvanishing Jacobian for all $a(\zeta)\in\mathcal{S}_n$. Therefore, a change of variables from $Q$ to $\mathbf{a}=[a_0,a_1,\dots,a_n]^\top$ in (\ref{dual}) is well defined, and the optimization of \eqref{dual}  can be transformed into     that of  Theorem \ref{propos_quasiNewton}   with $P(\zeta)=b(\zeta)b(\zeta^{-1})$; i.e.
\begin{equation}\label{opt_a}
\min_{\quad \ab\,;\, a(\zeta)\in\mathcal{S}_n\,}\, \mathbb{J}_P(\mathbf{a})  \; 
\end{equation}
where $\mathbb{J}_P(\mathbf{a})$ was defined in \eqref{func_a}.      Note, as observed by \cite{En},  that unlike the original optimization problem, the feasible set $\mathcal{S}_n$ is not convex in the $\mathbf{a}$-parameterization; nevertheless  the smooth bijection $T$ will map minima into minima.

\begin{thm}\label{periodThm}
Assume feasibility; then the optimization problem $(\ref{opt_a})$ has a unique stationary point   $\hat{\mathbf{a}}$ in $\mathcal{S}_n$ such that $T(\hat{a})=\hat{Q}(\zeta)$. The function $\mathbb{J}_P(\mathbf{a})$ is locally strictly convex in a neighborhood of such~$\hat{\mathbf{a}}$ which is indeed a minimum.  
\end{thm}
\begin{proof}The proof of the first part is essentially the same as that of  of Propositions 3.3 and 3.6 in \cite{En} and will not be repeated here.
The proof of local strict convexity will follow by showing that the Hessian of  $\mathbb{J}_P(\mathbf{a})$ is positive definite at $\hat{\ab}$ which is the content of  the following lemma.  

\begin{lemma}
The gradient of $\mathbb{J}_P(\mathbf{a})$ is related to $\nabla\mathbb{J}_P(Q)$ by the formula
\begin{equation}\label{grdt_relate}
\nabla\mathbb{J}_P(\mathbf{a})=\mathbf{J}^\top\nabla\mathbb{J}_P(Q)
\end{equation}
where $\mathbf{J}$ is the  Jacobian $\left [\frac{\partial q_i}{\partial a_j}\right]_{i,j=0,\ldots,n}$ which has the expression
\[\mathbf{J}=
\left[\begin{array}{ccccc}
a_0 & a_1 & \dots & a_n \\
a_1 & \dots & a_n & 0 \\
\vdots & \iddots &  & \vdots \\
a_n & 0 & \dots & 0
\end{array}\right]+
\left[\begin{array}{ccccc}
a_0 & a_1 & \dots & a_n \\
0 & a_0 & \dots & a_{n-1} \\
\vdots &  & \ddots & \vdots \\
0 & \dots & 0 & a_0
\end{array}\right] 
\]
and is non-singular, while the Hessian of $\mathbb{J}_P(\mathbf{a})$ is related to $\nabla^2\mathbb{J}_P(Q)$ by
\[\nabla^2\mathbb{J}_P(\mathbf{a})=\mathbf{J}^\top\nabla^2\mathbb{J}_P(Q)\mathbf{J}+\mathbf{R}\]
where
\begin{equation}\label{Hess_resi}
\mathbf{R}=c_0\mathbf{I}+\mathbf{T}_n-[\tilde{c}_0\mathbf{I}+\mathbf{T}_n(\mathbf{a})]
\end{equation}
with $\mathbf{T}_n(\mathbf{a})$ defined in Lemma $\ref{propos_grdt}$.
\end{lemma}
\begin{proof}
By the chain rule one has
\begin{equation}
\frac{d\mathbb{J}_P(\mathbf{a})}{d\mathbf{a}}=\frac{d\mathbb{J}_P(Q)}{d\mathbf{q}}\frac{d\mathbf{q}}{d\mathbf{a}}=\nabla\mathbb{J}_P(Q)^\top\mathbf{J}.
\end{equation}
Using the convention of writing the gradient as a column vector, one obtains (\ref{grdt_relate}) by taking a transpose.
The expression for $\mathbf{J}$ is a consequence of spectral factorization by matching the coefficients, as indicated by the quadratic equation (\ref{coeff_mathch}) in the appendix.

For the second statement, applying the chain rule and product rule for the derivative, one has
\[\nabla^2\mathbb{J}_P(\mathbf{a})=\mathbf{J}^\top\nabla^2\mathbb{J}_P(Q)\mathbf{J}+\left[\frac{d}{da}\mathbf{J}^\top\right]\nabla\mathbb{J}_P(Q).\]
It is not difficult to check that the $j$'th column of $\mathbf{R}$ is given by
\[\left[\frac{\partial}{\partial a_j}\mathbf{J}^\top\right]\nabla\mathbb{J}_P(Q)
=\left[\frac{\partial}{\partial a_j}\mathbf{J}^\top\right](\mathbf{c}-\tilde{\mathbf{c}}),\quad j=0,1,\dots,n,\]
which   leads  precisely to the structure in (\ref{Hess_resi}).
\end{proof}

Now, at the minimum $\hat{Q}$, the Hessian of the convex  functional 	$\mathbb{J}_P(Q)$ 	is positive definite while the matrix  $\Rb$ in  \eqref{Hess_resi} is clearly zero for $\ab=\hat{\ab}$ as $\mathbf{T}_n(\hat{\mathbf{a}})=\mathbf{T}_n$. Therefore $\nabla^2\mathbb{J}_P( \hat{\ab})$ is positive definite.
\end{proof}

We now propose an iterative algorithm to compute the spectral factor $a(\zeta)$ for the CRCEP. Note that the  iteration $(\ref{fixed-point_iteration})$ does not necessarily respect the constraint in the optimization problem (\ref{opt_a}). Hence a spectral factorization to extract the outer spectral factor may be  necessary to ensure   feasibility at each step. This  is similar to the projection step onto the feasible set in the projected Newton's method \cite{Be,Du} for constrained convex optimization.

\noindent{\bf Algorithm}\; [\,Quasi-Newton descent  with spectral factorization]\label{alg_prd}  
	\begin{enumerate}
		\item Initialize $\mathbf{a}^{(0)}=\mathbf{a}_{\mathrm{ME}}$, e.g. as the output of the Levinson algorithm for the ordinary covariance extension. Set a threshold $\delta$   to decide convergence\;
		
		\item Iterate $\mathbf{a}^{(k+1)}=\sigma^2(\mathbf{a}^{(k)})\,\mathbf{T}_n^{-1}\mathbf{T}_\gamma(\mathbf{a}^{(k)})\mathbf{b}$\;
			
		\item Do spectral factorization  $a^{(k+1)}(z)a^{(k+1)}(z^{-1})$  to get the outer spectral factor \;
		
		\item If $\|\mathbf{a}^{(k+1)}-\mathbf{a}^{(k)}\|>\delta$, go to step 2.
	\end{enumerate}

\begin{thm}
	  The quasi-Newton descent algorithm   converges locally to a vector of  AR coefficients    $\hat{\mathbf{a}}\in\mathcal{S}_n$, defining  a periodic ARMA process $(\ref{ARMA_periodic})$  whose covariance matrix is a circulant extension of the data $\mathbf{T}_n$, i.e.  solves Problem \ref{ARMA_identify}.
\end{thm}

\begin{proof} By local strict convexity of $\mathbb{J}_P(\mathbf{a})$, assuming  $\mathbf{a}^{(k)}$ is close enough to the minimum, the algorithm will converge to a point $\hat{\ab}$   satisfying the equation  
\[	\frac{1}{2}\,\nabla\mathbb{J}_P(\hat{\ab})= \mathbf{T}_n\hat{\ab}-\mathbf{T}_\gamma(\hat{\ab})\mathbf{b} =0\,.
\]	
To check that $\hat{\ab}$ solves the moment equations, just note that  the term $\mathbf{T}_n\hat{\mathbf{a}}$ can  be written as $\mathbf{M}(\hat{\mathbf{a}})\mathbf{c}$, where
\[\mathbf{M}(\mathbf{a}) =\left[\begin{array}{ccccc}
	a_0 & a_1 & a_2 & \dots & a_n \\
	a_1 & a_2 & \dots & a_n & 0 \\
	a_2 & \dots & a_n & 0 & 0 \\
	\vdots & \iddots & \iddots &  & \vdots \\
	a_n & 0 & \dots & 0 & 0
	\end{array}\right]+
	\left[\begin{array}{ccccc}
	0 & 0 & 0 & \dots & 0 \\
	0 & a_0 & 0 & \dots & 0 \\
	0 & a_1 & a_0 & 0 & 0 \\
	\vdots & \vdots & \ddots & \ddots & \vdots \\
	0 & a_{n-1} & \dots & a_1 & a_0
	\end{array}\right],\]
is the so-called {\em Jury matrix} mentioned in \cite{DM} whose determinant is
	\[\prod_{i=1}^{n}\prod_{j=1}^{n}(1-r_ir_j),\]
	where $r_i$ is the $i$'th root of the polynomial $a(z)$. Hence $\mathbf{M}(\hat{\mathbf{a}})$ is   nonsingular if $\hat{a}(z)$ is a Schur polynomial. 
Consider  then the equation in the unknown $\hat{\mathbf{c}}$
	\begin{equation}
	\mathbf{M}(\hat{\mathbf{a}})\hat{\mathbf{c}}=\sigma^2\mathbf{T}_\gamma(\hat{\mathbf{a}})\mathbf{b}.
	\end{equation}
which is just the matrix version of \eqref{ARMAforc} with $\ab=\hat{\ab}$ and the corresponding $\gamma= \gamma(\hat{\ab}) $ fixed. This is a linear equation which 	has as unique solution the vector $\hat{\mathbf{c}}=\mathbf{c}$, whose components are exactly the first $n+1$ covariance lags of the periodic ARMA process (\ref{ARMA_periodic}).
\end{proof}

\section{Generalization to the vector case}\label{SecVector}

The results obtained in the previous sections are for   scalar processes.  Their generalization   to   multidimensional processes however does not require a special treatment   but only involves a rather straightforward change of notations. In this section, we shall restrict our attention to periodic processes, as this setting is relevant to our application to smoothing. However, extension to vector-valued processes defined on $\mathbb{Z}$ is also straightforward. 

\subsection{Multidimensional CRCEP and the vector ARMA model}
Following \cite{LMPcirculant}, we present below a formulation of the multivariable circulant extension problem where $P(\zeta)$ is a scalar polynomial.  
\begin{prob}\label{CRCEP_vec}
	  Suppose we are given a scalar pseudo-polynomial
	\begin{equation}
	P(\zeta)=\sum_{k=-n}^{n}p_k\zeta^{-k}
	\end{equation}  
	which takes positive values on the discrete unit circle $\mathbb{T}_{2N}:=\{\zeta_k=e^{ik\pi/N}\ |\ k=-N+1,\dots,N\}$
	and covariance matrices $C_k:=\mathbb{E}\{y(t+k)y(t)^\top\}\in\mathbb{R}^{m\times m},\ k=0,1,\dots,n,$ of a certain stationary process $\{y(t)\}$,  such that the block-Toeplitz matrix
	\begin{equation}\label{MToepl}
	\mathbf{T}_n=\bmat
	C_0 & C_1 & \dots & C_n \\
	C_1^\top & C_0 &   & \vdots \\
	\vdots &   & \ddots & C_1  \\
	C_n^\top & \dots & C_1^\top & C_0   \emat	
\end{equation}	is positive definite. We want to determine an $m\times m$ matrix pseudo-polynomial 
	\begin{equation}
	Q(\zeta)=\sum_{k=-n}^{n}Q_k\zeta^{-k}
	\end{equation} such that $\Phi(\zeta):=Q(\zeta)^{-1}P(\zeta)$ is a   spectral density satisfying the moment conditions
	\[\int_{-\pi}^{\pi}e^{ik\theta}\Phi(e^{i\theta})d\nu=C_k,\quad k=0,1,\dots,n.\]
\end{prob}
In the ARMA formulation, we consider an $m$-dimensional stationary periodic process $\{y(t)\}$ described by an innovation  unilateral ARMA model driven by a $m$-dimensional  white noise $\{w(t)\}$
\begin{equation}\label{ARMA_prd_vec}
	\sum_{k=0}^{n}A_ky(t-k)=\sum_{k=0}^{n}b_kw(t-k),\quad t\in\mathbb{Z}_{2N}
\end{equation}
where $\{A_k\in\mathbb{R}^{m\times m}\}$ and $\{b_k\in\mathbb{R}\}$ are   coefficients of matrix Schur polynomials\footnote{Here the property of  a matrix polynomial  in the indeterminate $z^{-1}$  being Schur    is that all the roots of
$z^{mn}\det A(z)=0$ should lie inside the unit disc.} and $\mathbb{E}[w(t)w(t)^\top]=D>0$. Similar to the scalar case, we need to impose the periodic boundary condition (\ref{BC_prd}).
Introducing  block vector notations
\[  \yb=\bmat y(-N+1) \\ \vdots \\ y(N)\emat, \qquad
 \mathbf{w}=\bmat w(-N+1) \\ \vdots \\ w(N) \emat  \quad \in\mathbb{R}^{2mN}\]
and letting $\mathbf{D}:=\mathbb{E}\{\mathbf{w}\mathbf{w}^\top\}=I_{2N}\otimes D$, where $\otimes$ denotes the Kronecker product, then (\ref{ARMA_prd_vec}) can be written compactly as a matrix-vector equation
\begin{equation}\label{ARMA_mat_vec}
\mathbf{A}\mathbf{y}=\mathbf{B}\mathbf{w},
\end{equation}
in which $\mathbf{A}$ and $\mathbf{B}$ are $2mN\times2mN$ nonsingular lower-triangular block-circulant matrices of bandwidth $n$
\begin{equation}\label{Circ_A}
\mathbf{A}=\mathrm{Circ}\{A_0,A_1,\dots,A_n,0,\dots,0\},
\end{equation}
\begin{equation}\label{Circ_b}
\begin{split}
\mathbf{B} & =\mathrm{Circ}\{b_0I_m,b_1I_m,\dots,b_nI_m,0,\dots,0\} \\
 & =\mathrm{Circ}\{b_0,b_1,\dots,b_n,0,\dots,0\}\otimes I_m.
\end{split}
\end{equation}

\begin{prob}\label{ARMA_identify_vec}
Suppose we are given the MA coefficients $\{b_k\}$ of a Schur polynomial and $n+1$ real $m\times m$ matrices $C_0,C_1,\dots,C_n$,  such that the block-Toeplitz matrix \eqref{MToepl} is positive definite. We want to determine the matrix coefficients $\{A_k\}$ such that the first $n+1$ covariance matrices of the periodic process $\{y(t)\}$ match the sequence $\{C_k\}$.
\end{prob}
As in the scalar case, a solution of Problem \ref{ARMA_identify_vec}   solves Problem \ref{CRCEP_vec}. Actually, taking covariance on both sides of the equation (\ref{ARMA_mat_vec}), we   have
\[\mathbf{A}\mathbf{\Sigma}\mathbf{A}^\top=\mathbf{B}\mathbf{D}\mathbf{B}^\top,\]
where $\mathbf{\Sigma}:=\mathbb{E}\{\mathbf{y}\mathbf{y}^\top\}$. Taking inverses and rearranging terms, we obtain
 \begin{align}\label{VectCov}
\mathbf{\Sigma}  =\mathbf{A}^{-1}\mathbf{BDB^\top A^{-\top}}  &=(\mathbf{A^\top}\mathbf{D}^{-1}\mathbf{A})^{-1}\mathbf{B^\top B} \notag \\
& =\mathbf{Q}^{-1}\mathbf{P},
\end{align}
where $\mathbf{Q}:=\mathbf{A^\top}\mathbf{D}^{-1}\mathbf{A}$, $\mathbf{P}=\mathbf{B^\top B}$. The second equality follows from the observation that the block-circulant matrix $\mathbf{B}$ or $\mathbf{B}^\top$ commutes with other block-circulant matrices due to its special structure (\ref{Circ_b}).

\subsection*{Spectral representation and the vector Yule-Walker equation }
Directly from the time-domain representation \eqref{ARMA_prd_vec} we obtain   a spectral representation for the process $\{y(t)\}$
\begin{equation}
\hat{y}(\zeta)=A(\zeta )^{-1}b(\zeta)\hat{w}(\zeta),
\end{equation}
where $A(z)=\sum_{k=0}^{n}A_kz^{-k}$, and a discrete spectral density
\begin{equation}
\Phi(\zeta)=\left[A(\zeta^{-1})^\top D^{-1}A(\zeta)\right]^{-1}b(\zeta)b(\zeta^{-1}),\quad \zeta\in\mathbb{T}_{2N}.
\end{equation}
Next, introducing the impulse response
\begin{equation}
\Gamma_t:=\sum_{k=-N+1}^{N}\zeta_k^tA(\zeta_k)^{-1}b(\zeta_k)\frac{1}{2N},
\end{equation}
  we   obtain a representation for the process in terms of the input white noise as
\begin{equation}\label{one_sided_vec}
y(t)=\sum_{s=-N+1}^{N}\Gamma_{t-s}\, w(s).
\end{equation}
which can be rewritten in terms of the block-circulant matrix
\begin{equation}
\mathbf{\Gamma}=\mathrm{Circ}\{\Gamma_0,\Gamma_1,\dots,\Gamma_N,\Gamma_{-N+1},\dots,\Gamma_{-1}\},
\end{equation}
  as a matrix-vector product
\begin{equation}
\mathbf{y}=\mathbf{\Gamma w}\,.
\end{equation}
As in the scalar case,   block-circulant matrices lead to more compact notations. By following the same steps  as in Sec. \ref{covextsec} we obtain  the   relations
 \begin{equation*}
\mathbf{\Gamma}  =\mathbf{A}^{-1}\mathbf{B}, \quad
\mathbf{\Sigma}  =\mathbf{\Gamma D\Gamma^\top}, \quad
\mathbf{A\Sigma}  =\mathbf{BD\Gamma^\top}.
\end{equation*}
Combining the model equation (\ref{ARMA_prd_vec}) with the one-sided representation (\ref{one_sided_vec}), we obtain the vector analog of \eqref{ARMAforc}
\begin{equation}
\sum_{k=0}^{n}A_kC_{j-k}=D\sum_{k=0}^{n}b_k\Gamma_{k-j}^\top,\quad j=0,1,\dots,n,
\end{equation}
or, the   matrix equation
\begin{equation}\label{charac_A}
A\mathbf{T}_n=D B\mathbf{T}_\Gamma(A),
\end{equation}
where $A,B\in \Rbb^{m\times m(n+1)}$ are the  AR and MA matrix coefficients 
\[
A=\left[\begin{array}{cccc}
A_0 & A_1 & \dots & A_n
\end{array}\right],\quad
B=\left[\begin{array}{ccc}
b_0I_m  & \dots & b_nI_m
\end{array}\right]\]
and
\[
\mathbf{T}_\Gamma(A)  =\bmat 
\Gamma_0^\top & \Gamma_{-1}^\top & \dots & \Gamma_{-n}^\top \\
\Gamma_1^\top & \Gamma_0^\top & \dots & \Gamma_{-n+1}^\top \\
\vdots & \vdots & \ddots & \vdots \\
\Gamma_n^\top & \Gamma_{n-1}^\top & \dots & \Gamma_0^\top  
\emat \,.
\]
To solve \eqref{charac_A} we propose the iterative scheme  for the AR matrix coefficients
\begin{equation}\label{fixed_point_vec}
A^{(k+1)}=D\mathbf{B}\mathbf{T}_\Gamma(A^{(k)})\mathbf{T}_n^{-1},
\end{equation}
with $A^{(0)}$ initialized e.g. with the output of the Levinson-Whittle algorithm \cite{Whittle-63} for the data $\{C_k\}$, and the scaling matrix $D$   given by
\begin{equation}\label{D_in_A}
D(A):=\left(\sum_{k=0}^{n}A_kC_{-k}\right)\left(\sum_{k=0}^{n}b_k\Gamma_k^\top\right)^{-1}.
\end{equation}
The following proposition is a   generalization of Theorem \ref{propos_quasiNewton}. The proof is omitted as it is similar to that already given in the scalar case except for  more cumbersome notation.
\begin{prop}
	The  iteration $(\ref{fixed_point_vec})$ can be interpreted as a quasi-Newton step with a scaling matrix $D$ for the minimization of the function 
	\begin{equation}
\mathbb{J}_P(A)=\mathrm{tr}(A\mathbf{T}_nA^\top)-\int_{-\pi}^{\pi}b(e^{i\theta})b(e^{-i\theta})\log\det[A(e^{i\theta})^\top A(e^{-i\theta}) ]d\nu 
	\end{equation} 
subject to $A(\zeta)\in\mathcal{S}_n$ where $\mathcal{S}_n$ is the set of Schur matrix polynomials such that $A_0$ is lower-triangular with positive diagonal entries. 
\end{prop}
 The algorithm in the vector case is essentially the same as that in Sect. \ref{alg_prd} if we replace the scalar quantities with their vector counterparts. Again, we can always make $A_0=I_m$ by rescaling.

\section{Smoothing for periodic ARMA models}\label{Smoothing}
Consider the following problem. We have  a stationary vector signal $\{x(t)\}$ observed on the finite interval $[-N+1,N]$, the  observation channel being  described by the linear equation
\begin{equation}\label{obs_proc}
y(t)=Cx(t)+v(t),\qquad t\in \, [-N+1,\,N]
\end{equation}
where $\{v(t)\}$ is a stationary white noise with a known covariance matrix $R=R^\top>0$, independent of $\{x(t)\}$. We want to compute the smoothed estimate $\{\hat{x}(t)\}$ given the finite chunk of observations,
\begin{equation}\label{smoothed_proc}
\hat{x}(t):=\mathbb{E}\{x(t)|\ y(s),\ s\in[-N+1,N]\}
\end{equation}
with $N$ finite. The right-hand side of (\ref{smoothed_proc}) is the orthogonal projection onto the Hilbert space of random variables spanned by the components of $y(s),\ s\in[-N+1,N]$. We shall assume that the process $\{x(t)\}$ admits a bilateral ARMA model description of order $n$ on the interval $[-N+1,N]$. This description could be an approximation of an original stationary model for $x(t)$, say a Gauss-Markov model on the integer line $\Zbb$, obtained by matching a certain number of covariances. Let the bilateral model be
\[\sum_{k=-n}^{n}Q_k\,x(t-k)=\sum_{k=-n}^{n}P_k\, e(t-k),\]
with $\{e(t)\}$ the conjugate process. We shall use vector notations for the finite-interval restrictions of the underlying periodic processes as in \eqref{vecy}. The equivalent circulant model for the vector $\xb$ has the form
$$
\Qb\, \xb = \Pb \, \eb
$$
where $\Qb$ and $\Pb$ are $n$-banded positive block-circulants with elements $\{Q_k\}$ and $\{P_k \}$. Because of the orthogonality of $\xb$ to its conjugate process,   the covariance   $\Sigmab:=\E \xb\xb^{\top}$    has the expression
\begin{equation}\label{cov_x}
\mathbf{\Sigma}=\mathbf{Q}^{-1}\mathbf{P}\in\mathbb{R}^{2N\times2N}.
\end{equation}
which necessarily has a block-circulant structure.  With these data at hand we proceed to compute  the solution of the smoothing problem.    The  procedure is inspired to one     for reciprocal processes described   in \cite[Section VI]{Levy-F-K-90}. 

Write the observation equation   in   vector notation as
\begin{equation*}
\mathbf{y}=\mathbf{C}\mathbf{x}+\mathbf{v},
\end{equation*}
where $\Cb=\textrm{diag}\{C,\dots,C\}$. Since both $\xb$ and the white noise $\vb$ are periodic on $[-N+1,N]$ such is $\yb$. Then use  the standard one-shot solution for the minimum variance Bayesian estimate $\hat{\mathbf{x}}$ e.g. \cite[p. 29]{LPBook}  to get the relation
\begin{equation}\label{MVE}
(\mathbf{\Sigma}^{-1}+\mathbf{C}^\top\mathbf{R}^{-1}\mathbf{C})\hat{\mathbf{x}}=\mathbf{C}^\top\mathbf{R}^{-1}\mathbf{y},
\end{equation}
Substituting  (\ref{cov_x}) into (\ref{MVE}), the matrix on the left-hand side becomes
\begin{equation}
\mathbf{P}^{-1}\mathbf{Q}+\mathbf{C}^\top\mathbf{R}^{-1}\mathbf{C}=\mathbf{P}^{-1}(\mathbf{Q}+\mathbf{P}\mathbf{C}^\top\mathbf{R}^{-1}\mathbf{C}).
\end{equation}
Then define 
\begin{equation}\label{Qhat}
\hat{\mathbf{Q}}:=\mathbf{Q}+\mathbf{P}\mathbf{C}^\top\mathbf{R}^{-1}\mathbf{C}\,, 
\end{equation}
which is a positive-definite  block-circulant since $\mathbf{C}^\top\mathbf{R}^{-1}\mathbf{C}$ is a block-diagonal matrix with positive-semidefinite blocks. In fact $\hat{\mathbf{Q}}$ is bilaterally banded of bandwidth $n$ since such are both summands in the right hand member of \eqref{Qhat}. Then (\ref{MVE}) is equivalent to
\begin{equation}\label{one-shot}
\hat{\mathbf{Q}}\hat{\mathbf{x}}=\mathbf{P}\mathbf{C}^\top\mathbf{R}^{-1}\mathbf{y}:=\hat{\yb}.
\end{equation}
In order to carry out a two-sweep smoothing procedure in the style of the Rauch-Striebel-Tung smoother \cite{Rauch-S-T-65}, we first perform a banded matrix factorization
\[\hat{\mathbf{Q}}=\hat{\Ab}\hat{\Ab}^\top,\]
where
\begin{equation}
	\hat{\mathbf{A}}=\mathrm{Circ}\{\hat{A}_0,\hat{A}_1,\dots,\hat{A}_n,0,\dots,0\}.
\end{equation}
As discussed in Appendix \ref{unilateralARMAsec}, such a factorization is possible if $N$ is taken large enough and can be computed in the spectral domain by standard matrix polynomial factorization algorithms, see e.g. \cite{Rissanen-73}. Then, given $\hat{\Ab}$ and $\hat{\yb}$, to compute the solution to \eqref{one-shot} we first perform a forward sweep described by 
\begin{equation}\label{ForwS}
	\hat{\mathbf{A}}\mathbf{z}=\hat{\yb},
\end{equation}
and then a backward sweep 
\begin{equation}\label{BackS}
	\hat{\mathbf{A}}^\top\hat{\mathbf{x}}=\mathbf{z}\,.
\end{equation}
The two sweeps can be implemented by a forward and a backward recursive algorithm described by    unilateral AR  models. To this  end we need to attach to them explicit boundary values $\hat{x}(-N+1),\hat{x}(-N+2),\dots,\hat{x}(-N+n)$ and $\hat{x}(N-n+1),\dots,\hat{x}(N)$ extracted from  the process $\{\hat x(t)\}$ which we assume are given. Due to the   lower block-triangular structure of $\hat{\Ab}$, the first equation of the forward sweep can be written as
\begin{equation}
\hat{A}_0 z(-N+1)= -\sum_{i=1}^{n}\hat{A}_i z(N-i+1)+\hat{y}(t)\,,
\end{equation}
which needs to be initialized with the boundary values $z(N-n+1),\,z(N-n +2),\,\ldots,z(N)$. These values can be obtained by solving for $z$ the last $n$ block equations in the backward sweep \eqref{BackS} since only the boundary values at two ends of $\hat{\mathbf{x}}$ are involved there due to the banded upper-triangular block-circulant structure of $\hat{\Ab}^{\top}$.

The forward sweep starts by computing the boundary values $z(N-n+1),\dots,z(N)$. After    these $n$ endpoint boundary values are available, the recursion  for $\zb$ can be implemented by the scheme
\begin{equation}
z(t)=\hat{A}_0^{-1}\left[\hat{y}(t)-\sum_{i=1}^{n}\hat{A}_i z(t-i)\right],\qquad t\in[-N+1,N-n].
\end{equation}
One should notice that in this notation, we impose implicitly that $z(-N)=z(N),\dots,z(-N-n+1)=z(N-n+1)$.
 The  backward sweep then proceeds  by using
\begin{equation}
	\hat{x}(t)=\hat{A}_0^{-\top}\left[z(t)-\sum_{i=1}^{n}\hat{A}_i^\top\hat{x}(t+i)\right],\quad t\in[-N+n+1,N-n],\,
\end{equation}
which is initialized with the known terminal boundary values $\hat{x}(N-n+1),\hat{x}(N-n+2),\dots,\hat{x}(N)$.

There is also a dual factorization which would lead to a backward-forward sequence of sweeps but we shall not insist on this point.

\subsection*{A numerical example}
Suppose that we are given a stationary stochastic system in the form of a state-space model
\begin{equation}\label{state-space}
\left\{ \begin{array}{ll}
x(t+1) & =Ax(t)+w(t)\\
y(t) & =Cx(t)+v(t)\\
\end{array} \right.
\end{equation}
where
\begin{equation}
A=\left[\begin{array}{cc}
0.9 & -0.3 \\
0.3 & 0.9
\end{array}\right],\quad
C=\left[\begin{array}{cc}
1 & 2 \\
1 & 0
\end{array}\right]
\end{equation}
are constant matrices, and the processes $\{w(t)\}$ and $\{v(t)\}$ are uncorrelated Gaussian white noises with unit variance. The eigenvalues of $A$ are $0.9\pm0.3i$ with a modulus $0.9487$.

We want to compute the smoothed process (\ref{smoothed_proc}). To do this, we first build a periodic ARMA model of order $n=1$ to approximately describe the process $\{x(t)\}$ on a finite interval by matching the first two stationary state covariances $C_0,C_1$.  The period of interest is set as $2N=50$ and the MA parameters are chosen (quite arbitrarily) as $b_0=1,b_1=0.5$. The unilateral ARMA model looks like
\begin{equation}\label{ARMA_1_1}
A_0x(t)+A_1x(t-1)=b_0w(t)+b_1w(t-1),
\end{equation}
and the AR parameters are computed with a variation of the algorithm of Sect. \ref{alg_prd} adapted to the vector case with $A_0$ scaled to identity and
\begin{equation*}
A_1=\left[\begin{array}{cc}
-0.8609 & 0.2989\\
-0.2989 & -0.8609
\end{array}\right],
\quad D=\left[\begin{array}{cc}
0.8122 & 0 \\
0 & 0.8122 \\
\end{array}\right].
\end{equation*}
Given the model (\ref{ARMA_1_1}) and the observation process (\ref{obs_proc}), the two-sweep smoothing algorithm described in the previous part can be implemented. The two components of the smoothed process $\hat{x}(t)$ computed using the approximate periodic model (\ref{ARMA_1_1}) are shown in figures~\ref{fig1} and \ref{fig2}. The effect of smoothing is appreciable.

\begin{figure}[h!]
	\centering
	\includegraphics[width=0.9\textwidth]{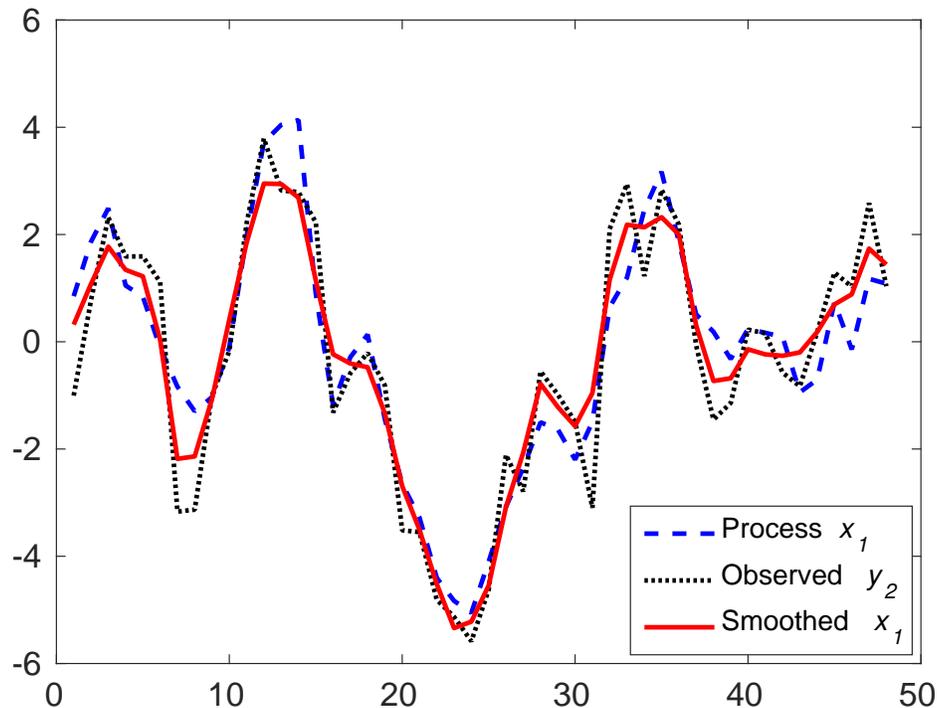}
	\caption{\label{fig1} Result of smoothing for $x_1$}
\end{figure}
\begin{figure}[h!]
	\centering
	\includegraphics[width=0.9\textwidth]{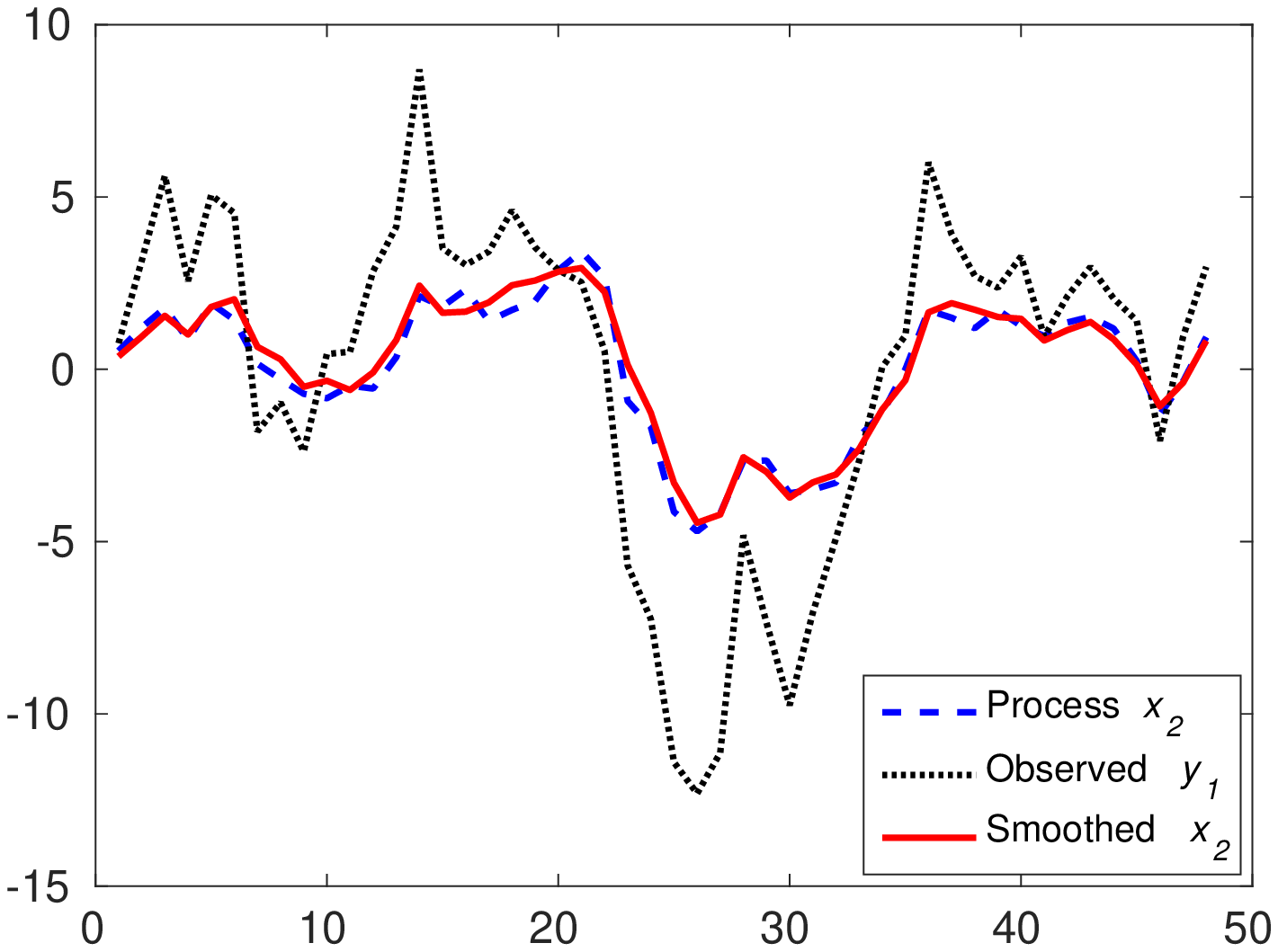}
	\caption{\label{fig2} Result of smoothing for $x_2$}
\end{figure}

\section{Conclusion}
We have developed a new iterative algorithm to solve the rational covariance extension problem, for both the ordinary and the periodic case, using the familiar Yule-Walker-type equations for the ARMA model. The results are also generalized to the vector case and used to    extract approximate stationary models on a finite interval. The procedure  works well for the finite-interval smoothing  and the resulting algorithm  is computationally cheaper than the standard Riccati-based smoother.

\begin{appendices}
\section{Harmonic analysis in $\Zbb_{2N}$ and stationary periodic vector processes}\label{AppA}

The discrete Fourier transform (DFT) $\script F$ maps a finite sequence   $\gb=\{ \gb_k; \,k= -N+1,\,\dots ,\, N\}$,  $ \gb_k \in  \mathbb{C}^m$, into a  sequence  of complex $m$-vectors
\begin{equation} \label{DFT}
\hat{\gb}(\zeta_{j}) := \sum_{k=-N+1 }^{N }\, \gb_k \zeta_{j}^{-k} \,,\qquad j=-N+1,-N+2, \ldots , N,
\end{equation}
where $\zeta_j:=e^{ij\pi/N}$. Here   the discrete variable $\zeta$ takes the $2N$  values $ \zeta_j$, $j=-N+1,\ldots,0,\ldots,N$ and  runs counterclockwise on the  discrete unit circle $\Tbb_{2N}$.   In particular, we have $\zeta_j=(\zeta_1)^j$  and $\zeta_{-k}=\overline{\zeta_k}$. The inverse DFT $\script F^{-1}$ is given by 
\begin{equation} \label{InvDFT}
\gb_k = \frac{1}{2N}\sum_{j=-N+1}^N\zeta_{j}^k \hat{\gb}(\zeta_j),\quad k =  -N+1,-N+2, \dots ,N,
\end{equation}
which can also be written as a Stieltjes integral
\begin{equation} \label{InvDFTmeas}
\gb_k = \int_{-\pi}^\pi e^{ik\theta}\hat{\gb}(e^{i\theta})  d\nu(\theta),\quad k =  -N+1,-N+2, \dots ,N,
\end{equation}
where  $\nu$ is a step function with steps $\frac{1}{2N}$ at each $\zeta_k$; i.e., 
\begin{equation}
\label{nu}
d\nu(\theta) =\sum_{j=-N+1}^N\delta(e^{i\theta}-\zeta_j)\frac{d\theta}{2N}.
\end{equation}
This makes the DFT   a unitary map from $\mathbb{C}^{m\times 2N}$ onto $L_{m}^2([-\pi,\pi],d\nu)$ so that for arbitrary   strings $\hb=\{\hb_k\}$, $\gb=\{\gb_k\}$, one has
\begin{equation}
\label{Plancherel}
\!\langle \gb,\hb\rangle_{\Cbb^{m\times 2N}}\!=\!\frac{1}{2N}\!\!\!\sum_{k=-N+1}^N \hat{\gb}(\zeta_k)\hat{\hb}(\zeta_{-k})^*\!= \!\int_{-\pi}^\pi \!\hat{\gb}(e^{i\theta})\hat{\hb}(e^{i\theta})^*d\nu,
\end{equation}
which is {\em Plancherel's Theorem\/} for DFT.

Next consider  a zero-mean  stationary $m$-dimensional process $\{y(t)\}$  defined on $\Zbb_{2N}$; i.e., a stationary process defined on a finite interval $[-N+1,\,N]$ of the integer line $\Zbb$ and extended to all of $\Zbb$ as a periodic stationary process with  period $2N$.  Let $C_{-N+1},C_{-N+2},\dots,C_{N}$ be the $m\times m$ covariance lags $C_k:=\E\{ y(t+k)y(t)^*\}$, and define its discrete Fourier transformation
\begin{equation} \label{c2Phi}
\Phi(\zeta_{j}) := \sum_{k=-N+1 }^{N }\, C_k \zeta_{j}^{-k} \,,\qquad j=-N+1,\dots , N,
\end{equation}
which is a positive, Hermitian matrix-valued function of $\zeta$. Then, as seen from \eqref{InvDFT} and \eqref{InvDFTmeas}, 
\begin{equation} \label{Phi2c}
C_k = \frac{1}{2N}\sum_{j=-N+1 }^{N  }\zeta_{j}^k \Phi(\zeta_j)=\int_{-\pi}^\pi e^{ik\theta}\Phi(e^{i\theta})  d\nu,\quad k \in\Zbb_{2N}
\end{equation}
The $m\times m$ matrix  function $\Phi$ is the {\em spectral density\/} of the vector process $y$. In fact, let
\begin{equation}
\label{yDFT}
\hat{y}(\zeta_k):= \sum_{t=-N+1}^N y(t)\zeta_k^{-t}, \quad k=-N+1,\dots, N,
\end{equation}
be the discrete Fourier transformation of the process $y$. Since $\frac{1}{2N}\sum_{t=-N+1}^N (\zeta_k\zeta_\ell^*)^t =\delta_{k\ell}$, the random variables \eqref{yDFT} are uncorrelated, and 
\begin{equation}
\label{yhatyhat}
\frac{1}{2N}\E\{ \hat{y}(\zeta_k)\hat{y}(\zeta_\ell)^*\}=\Phi(\zeta_{k})\delta_{k\ell}.
\end{equation}
This yields a spectral representation of $y$ analogous to the usual one, namely
\begin{equation}
\label{ }
y(t)=\frac{1}{2N}\sum_{k=-N+1}^N \zeta_k^t\,\hat{y}(\zeta_k)=\int_{-\pi}^\pi e^{ik\theta}d\hat{y}(\theta),
\end{equation}
where $d\hat{y}(\theta):=\hat{y}(e^{i\theta})d\nu(\theta)$.

\subsection*{Block-circulant matrices}\label{CircMat}
  In the multivariable circulant rational covariance extension problem we consider {\em Hermitian\/} circulant  matrices 
\begin{equation}
\label{M_C}
\Mb:=\Circ\{ M_0,M_1,M_2,\dots, M_N,M_{N-1}^*,\dots,M_2^*,M_1^*\},
\end{equation}
where, by periodicity $M_{-k}=M_k^*$ so that they can be represented in form
\begin{equation}
\label{S2C}
\Mb =\sum_{k=-N+1}^N S^{-k}\otimes M_k, \quad 
\end{equation}
where $\otimes$ is the Kronecker product and $S$ is the nonsingular $2N\times 2N$  cyclic shift matrix
\begin{equation}
\label{S}
S := \left[\begin{array}{cccccc}0 & 1 & 0 & 0 & \dots & 0 \\0 & 0 & 1 & 0 & \dots & 0 \\0 & 0 & 0 & 1 & \dots & 0 \\\vdots & \vdots & \vdots & \ddots & \ddots & \vdots \\0 & 0 & 0 & 0 & 0 & 1 \\1 & 0 & 0 & 0 & 0 & 0\end{array}\right].
\end{equation}
The $m\times m$ matrix  pseudo-polynomial
\begin{equation}
\label{symbol}
M(\zeta)=\sum_{k=-N+1}^N M_k \zeta^{-k}, \quad M_{-k}=M_k^*
\end{equation}
is called the {\em symbol\/} of $\Mb$. Let $\Sb=S\otimes I_m$ be the $2mN\times 2mN$ cyclic shift matrix which satifies the cyclic relations   $\Sb^{2N}=\Sb^0=\Ib :=I_{2mN}$, and 
\begin{equation}
\label{Sbpowers}
\Sb^{k+2N}=\Sb^k, \qquad \Sb^{2N-k}=\Sb^{-k}=(\Sb^k)\Tr .
\end{equation}
It i snot hard to show that
\begin{equation}
\label{circulantcondition}
\Sb\Mb \Sb^*=\Mb,
\end{equation}
 is both necessary and sufficient for $\Mb$ to be block-circulant. With $\gb:=\big(\gb_{-N+1}\Tr,\gb_{-N+2}\Tr,\dots,\gb_N\Tr\big)\Tr$, we have
\begin{equation}\label{ }
[\Sb\gb]_k=\gb_{k+1}, \quad k\in\mathbb{Z}_{2N}. 
\end{equation}
Then, in view of \eqref{DFT},  $\zeta {\script F}(\gb)(\zeta)={\script F}(\Sb\gb)(\zeta)$,
from which it follows that
\begin{equation}
\label{Cg}
{\script F}(\Mb\gb)(\zeta)=M(\zeta){\script F}(\gb)(\zeta),
\end{equation}
where the $m\times m$ matrix fuction $M(\zeta)$ is the symbol \eqref{symbol} of the block-circulant matrix $\Mb$.
An important property of block-circulant matrices is that they can be block-diagonalized by the discrete Fourier transform. Using this fact  it follows from \eqref{Cg} that  
\begin{displaymath}
\Sb\Mb^{-1}\Sb^*=\Mb^{-1}.
\end{displaymath}
Consequently, $\Mb^{-1}$ is also a block-circulant  matrix with symbol $M(\zeta)^{-1}$. In general, in view of the block-circulant property \eqref{S2C} and \eqref{Sbpowers}, quotients of symbols are themselves pseudo-polynomials of degree at most $N$ and hence symbols. More generally, if $\Ab$ and $\Bb$ are block-circulant matrices of the same dimension with symbols $A(\zeta)$ and $B(\zeta)$ respectively, then $\Ab\Bb$ and $\Ab+\Bb$ are block-circulant matrices with symbols $A(\zeta)B(\zeta)$ and $A(\zeta)+B(\zeta)$, respectively.  In fact, block-circulant matrices of a fixed dimension form an algebra, and the DFT is an {\em algebra   homomorphism} of the set of  circulant matrices  onto the   pseudo-polynomials of degree at most $N$ in the variable $\zeta \in \Tbb_{2N}$.

\section{Unilateral ARMA models and spectral factorization}\label{unilateralARMAsec}
Analysis in this part of the appendix is confined in the scalar case  for simplicity. Generalization to vector processes is straightforward.
As explained in Section~\ref{covextsec}, a periodic process $y$ has a discrete spectrum $\Phi(\zeta)$   defined 
only in the discrete points of $ \mathbb{T}_{2N}:=\{\zeta_{-N+1},\zeta_{-N+2},\dots,\zeta_N\}$. Since $\Phi$ takes positive values on $\mathbb{T}_{2N}$, there are  trivial discrete factorizations
\begin{equation}
\label{discretefactor}
\Phi(\zeta_k)=W(\zeta_k)W(\zeta_k)^*\quad k=-N+1,\dots, N\,,
\end{equation}
so that, defining
\begin{displaymath}
W_k= \frac{1}{2N}\sum_{j=-N+1}^N W(\zeta_j)\zeta_j^k, \quad k=-N+1,\dots, N,
\end{displaymath}
we can write \eqref{discretefactor} in the form 
\begin{equation}
\label{Phi(zeta)}
\Phi(\zeta)=W(\zeta)W(\zeta)^*.
\end{equation}
where $W(\zeta)$ is the discrete Fourier transform
\begin{displaymath}
W(\zeta)=\sum_{k=-N+1}^N W_k\zeta^{-k}.
\end{displaymath}

Formally substituting the variable  $z\in \mathbb{T}$ in place of $\zeta$ in $W$, we obtain a spectral factorization equation on the whole unit circle, 
\begin{equation}
\label{Phi(z)}
\tilde\Phi(z)=W(z)W(z)^*,\quad z\in\mathbb{T},
\end{equation}
  where $\tilde{\Phi}(z)$ must then be positive on $\Tbb$ and hence   a continuous spectral density which, frequency sampled with sampling interval $\frac{\pi}{N}$,  yields $\tilde\Phi(\zeta)=\Phi(\zeta)$ on $\mathbb{T}_{2N}$. This $\tilde{\Phi}(z)$ is a spectral density of a non-periodic stationary process  which  has the covariance lags
\begin{displaymath}
\tilde{c}_k=\int_{-\pi}^\pi e^{ik\theta}\tilde\Phi(e^{i\theta})\frac{d\theta}{2\pi}, \quad k=0,1,\dots,n ,
\end{displaymath}
 differing from $c_0,c_1,\dots,c_n$. However, setting $\Delta\theta_j:=\theta_j-\theta_{j-1}$ where $e^{\theta_j}=\zeta_j$,   we see from \eqref{zetadefn} that $\Delta\theta_j=\pi/N$ and that the integral   with $\tilde\Phi$ fixed is approximated by the Riemann sum
\begin{displaymath}
c_k= \sum_{j=-N+1}^Ne^{ik\theta_j} \tilde\Phi(\zeta_j)\frac{\Delta\theta_j}{2\pi} 
\end{displaymath}
converging to $\tilde{c}_k$ for $k=0,1,\dots,n$ as $N\to\infty$. In particular,  $\tilde\Phi \in L^1(\Tbb)$ is a bona fide rational  spectral density which has a unique outer spectral factor $W(z)$; see, e.g., \cite{LPBook}.

Hence, since  $\Phi(\zeta)$ is the symbol of the circulant covariance matrix $\Sigmab$,   \eqref{Phi(zeta)} can be written in  matrix form as
\begin{equation}
\label{matrixWW*}
\Sigmab=\Wb\Wb^*,
\end{equation}
where $\Wb$ is the $2N\times 2N$ circulant matrix with symbol $W(\zeta)$.  As explained in   \cite{Pi-MTNS016}, in the discrete setting $W(\zeta)$ can take the form corresponding to the outer spectral factor in (\ref{Phi(z)})
\begin{equation}
\label{analyticW}
W(\zeta)=\sum_{k=0}^N W_k\zeta^{-k},
\end{equation}
which in turn corresponds to $\Wb$ being  {\em lower-triangular circulant}, i.e.,
\begin{equation}
\label{ }
\Wb=\Circ\{ W_0,W_1,\dots,W_N,0,\dots,0\}.
\end{equation}
Note that a lower-triangular circulant matrix is not really lower triangular as the circulant structure has to be preserved. 
Since $\Sigmab$ is invertible, then so is $\Wb$.   

Next define the periodic stochastic process $\{w(t),\, t=-N+1\dots,N\}$ for which $\wb = [w(-N +1),w(-N +2),\dots,w(N)]\Tr$ is given by
\begin{equation}
\label{wdefn}
\wb=\Wb^{-1}\yb.
\end{equation}
Then, in view of \eqref{matrixWW*}, we obtain $\E\{\wb\wb^*\}=\Ib$, i.e., the process $w$ is a periodic white noise process. Consequently we have the unilateral representation 
\begin{displaymath}
y(t)=\sum_{k=0}^N W_k\, w(t-k)
\end{displaymath}
in terms of white noise.

\subsection*{Factorizability of polynomials  and of   banded circulant matrices}\label{sec_factorizability}
We shall consider the spectral  factorization of a scalar pseudopolynomial $p(\zeta)\in\mathfrak{P}_+(N)$ of degree $n$ as
	\begin{equation}\label{circ_fact}
	p(\zeta)=a(\zeta)a(\zeta^{-1}),\quad \zeta\in\mathbb{T}_{2N},
	\end{equation}
where $ a(\zeta)=\sum_{k=0}^{n}a_k\zeta^{-k}. $

This problem is equivalent to factorization of symmetric 	$n$-banded positive-definite circulant matrices. Given such a matrix	\begin{equation}
	\mathbf{P}=\mathrm{Circ}\{p_0,p_1,\dots,p_n,0,\dots,0,p_n,p_{n-1},\dots,p_1\}\in\mathbb{R}^{2N\times 2N},
	\end{equation}
an equivalent question is under what condition does it admit a banded $2N\times 2N$ circulant factor
$ \mathbf{P}=\mathbf{AA^\top} $
with
	\begin{equation}
	\mathbf{A}=\mathrm{Circ}\{a_0,a_1,\dots,a_n,0,\dots,0\}.
\end{equation}

	Equation (\ref{circ_fact}) looks very much like the polynomial factorization in the complex variable $z$
	\begin{equation}\label{ordy_fact}
	p(z)=a(z)a(z^{-1}),
	\end{equation}
which is well-known to admit Schur (in particular outer) solutions if and only if $p(\zeta)\in \mathfrak{P}_+$, \cite{LPBook}.

Clearly 	If $p(z)$ is factorizable as in (\ref{ordy_fact}), then a corresponding factorization  (\ref{circ_fact}) will hold, since $\zeta$ is a restriction of $z$ to the set $\mathbb{T}_{2N}$ and the coefficients of the polynomial factors can be chosen to be the same. In particular to  the (unique)  Schur  factor there will  correspond a polynomial factor solution of \eqref{circ_fact} which we shall still call Schur. In fact, the converse is also true.

\begin{lemma}
		Asume $N>n$, then the pseudo-polynomial $p(\zeta)$ is factorizable if and only if the polynomial $p(z)$ is factorizable, and the outer factors have the same coefficients.
\end{lemma}
\begin{proof}
Assume \eqref{circ_fact} holds and define the pseudo-polynomial in $z$
	\[q(z)=\sum_{k=-n}^{n}q_kz^{-k}:=a(z^{-1})a(z)\,. \]
By matching the coefficients	it must hold that
	\begin{equation}\label{coeff_mathch}
	\left[\begin{array}{cccc}
	a_0 & a_1 & \dots & a_n \\
	& a_0 & \dots & a_{n-1} \\
	&  & \ddots & \vdots \\
	&  &  & a_0
	\end{array}\right]
	\left[\begin{array}{c}
	a_0 \\
	a_1 \\
	\vdots \\
	a_n
	\end{array}\right]
	=\left[\begin{array}{c}
	q_0 \\
	q_1 \\
	\vdots \\
	q_n
	\end{array}\right]
	\end{equation}
so that	
\begin{equation}\label{coeff_poly}
	\sum_{k=-n}^{n}(p_k-q_k)\zeta_j^{-k}=0,\quad j=-N+1,\dots,N.
	\end{equation}
	Viewing $\{p_k-q_k\}_{k=-n}^n$ as variables,   (\ref{coeff_poly}) is an overdetermined linear system of equations, with a $2N\times(2n+1)$ Vandermonde matrix as the coefficient matrix, which is apparently of full column rank. Thus, the solution vector must be zero, i.e.,
	\begin{equation}
	p_k=q_k,\quad k=-n,\dots,n.
	\end{equation}
	This implies that $p(z)$ admits a polynomial factorization (\ref{ordy_fact}).
\end{proof}		
	There are efficient algorithms to compute the outer polynomial factor directly without solving for all the roots of $p(z)$, as described in \cite{Rissanen-73} . 
 
 We now turn to the unilateral ARMA representation of a periodic process with a rational spectral density: 
\begin{displaymath}
\Phi(\zeta)=\frac{P(\zeta)}{Q(\zeta)}\,,\qquad \zeta \in \Tbb_{2N}
\end{displaymath}
  This representation clearly requires both $P(\zeta)$ and $Q(\zeta)$ to admit polynomial spectral factors  of finite degree $n<N$.  There is a difficulty here since,  $P(\zeta)$ and $Q(\zeta)$ can admit  polynomial spectral factors   if and only if their extension $P(z),\,Q(z)$ to the unit circle does so. In other words   positivity on the discrete set $\Tbb_{2N}$ must imply positivity as polynomial functions of $z\in \Tbb$. To this end we may provide the following criterion.

\begin{lemma}\label{positivity}
Let $p(\zeta)\in \mathfrak{P}_+(N)$ be  a polynomial of degree $n$. Then if $N$ is large enough, the extension of $p(\zeta)$  to the unit circle $p(z)\,;\,  z\in \Tbb$, must also be positive for all $z\in\Tbb$.
\end{lemma}
\begin{proof} For assume that for some $\zeta_0\in \Tbb$, $p(\zeta_0)<0$; then there must be an interval  neighborhood of $\zeta_0$ in $\Tbb$ having finite measure where $p(e^{i\theta})<0$. But if $N$ is large enough some $\zeta_k\in \Tbb_{2N}$ must belong to this neighborhood and  then $p(\zeta_k)$ must be negative which is impossible.
\end{proof}

The following   corollary    also holds for block circulant matrices which are considered in Section~\ref{Smoothing}. 

\begin{cor}\label{lem:bandedfactor}
If $N$ is large enough, a positive definite  Hermitian  circulant matrix $\Mb$ admits a factorization $\Mb=\Vb\Vb^*$, where $\Vb$ is a banded lower-triangular circulant matrix of order $n<N$, if and only if $\Mb$ is bilaterally banded of order $n$.
\end{cor}

The covariance matrix of a periodic process $\yb$ having a rational spectral density $\Phi(\zeta)= P(\zeta)/Q(\zeta)$ has the representation $\Sigmab=\Qb^{-1}\Pb$, where $\Qb$ and $\Pb$ are banded, positive definite, Hermitian, circulant matrices of order $n$ having symbols $Q(\zeta)$ and $P(\zeta)$. 
Hence, by Corollary~\ref{lem:bandedfactor}, for $N$ large enough there are factorizations
\begin{displaymath}
\Qb=\Ab\Ab^*\quad\text{and}\quad \Pb=\Bb\Bb^*,
\end{displaymath}
where $\Ab$ and $\Bb$ are banded lower-diagonal circulant matrices of order $n$. Consequently, $\Sigmab=\Ab^{-1}\Bb(\Ab^{-1}\Bb)^*$, i.e., 
\begin{equation}
\label{W=AinvB}
\Wb=\Ab^{-1}\Bb,
\end{equation}
which together with \eqref{wdefn} yields $\Ab\yb=\Bb\wb$, i.e., the unilateral ARMA model
\begin{equation}
\label{unilateralARMA}
\sum_{k=0}^n a_k y(t-k) = \sum_{k=0}^n b_k w(t-k).
\end{equation}
Since $\Ab$ is nonsingular, $a_0\ne 0$, and hence we can normalize by setting $a_0=1$.
In particular, if $\Pb=\Ib$, we obtain the AR representation
\begin{equation}
\label{unilateralAR}
\sum_{k=0}^n a_k y(t-k) =  w(t).
\end{equation}
Symmetrically, there is factorization
\begin{equation}
\label{matrixWbarWbar*}
\Sigmab=\bar{\Wb}\bar{\Wb}^*,
\end{equation}
where $\bar{\Wb}$ is upper-triangular circulant, i.e. the transpose of a lower-triangular circulant matrix, and a white-noise process
\begin{equation}
\label{wbardefn}
\bar\wb=\bar{\Wb}^{-1}\yb.
\end{equation}
Likewise there are factorizations 
\begin{displaymath}
\Qb=\bar\Ab\bar\Ab^*\quad\text{and}\quad \Pb=\bar\Bb\bar\Bb^*,
\end{displaymath}
where $\bar\Ab$ and $\bar\Bb$ are banded upper-diagonal circulant matrices of order $n$. This yields a backward unilateral ARMA model
\begin{equation}
\label{unilateralARMAbar}
\sum_{k=-n}^0 \bar a_k y(t-k) = \sum_{k=-n}^0 \bar b_k \bar w(t-k).
\end{equation}

These representations are useful in the smoothing problem for periodic systems of Sect. \ref{Smoothing}.

\end{appendices}
\begin{IEEEbiography}
{Giorgio Picci} (S'67ÐM'70ÐSM'91ÐF'94ÐLF'08) received the Dr.Eng. degree from the University of Padua, Padua, Italy, in 1967.
Currently, he is Professor Emeritus with the Department of Information Engineering, University of Padua, Padua, Italy. He has held several long-term visiting appointments with various American, Japanese, and European universities among which Brown University, MIT, the University of Kentucky, Arizona State University, the Center for Mathematics and Computer Sciences (CWI) in Amsterdam, the Royal Institute of Technology, Stockholm, Sweden, Kyoto University, and Washington University, St. Louis, MO, USA. He has been contributing to systems and control mostly in the area of modeling, estimation, and identification of stochastic systems and published over 150 papers and written or edited several books in this area. He has been involved in various joint research projects with industry and state agencies. Besides being a life Fellow of IEEE, he 
is a Fellow of IFAC and a foreign member of the Swedish Royal Academy of Engineering Sciences.  .
\end{IEEEbiography}
\begin{IEEEbiography}
{Bin Zhu} received a Bachelor's degree from Xi'an Jiaotong University, Xi'an, China in 2012 and a Master's degree from Shanghai Jiao Tong University, Shanghai, China in 2015, both in control science and engineering. He is now a Ph.D. student at the Department of Information Engineering, University of Padova, Padova, Italy.
His current research interest includes system identification, modeling, signal processing and estimation.
\end{IEEEbiography}


\begin{thebibliography}{99}

\bibitem{Be} D.~P.~Bertsekas, Projected Newton methods for optimization problems with simple constraints,	\emph{SIAM J.~Control and Optimization}, 20(2): 221--246, 1982.


\bibitem{BLGM1} C. I. Byrnes,  A. Lindquist, S.V. Gusev,  and A. 
V. Matveev, A complete parameterization of all positive rational extensions of a covariance sequence, {\em IEEE Trans. Aut. Contr.} {\bf AC-40} (1995) 1841-1857.

\bibitem{Byrnes-L-97}
C. I. Byrnes and A. Lindquist, On the partial stochastic realization problem, {\em IEEE Transactions on Automatic Control} {\bf AC-42} (1997), 1049--1069. 

\bibitem{BGuL} C. I. Byrnes,  S. V. Gusev,  and A. Lindquist,
`A convex optimization approach to the rational covariance
extension problem,
{\em SIAM J. Control and Opt.} {\bf 37} (1999), 211-229.

\bibitem{SIGEST} C. I. Byrnes,  S.V. Gusev,  and A. Lindquist,
From finite covariance windows to modeling filters: A convex optimization
approach, {\em SIAM Review} {\bf 43} (2001) 645--675.

\bibitem{BEL1} C. I. Byrnes,  P. Enqvist,  and A. Lindquist,
{\em Cepstral coefficients, covariance lags and pole-zero models for finite data
strings}, IEEE Trans.\ on Signal Processing {\bf SP-50} (2001), 677--693.

\bibitem{BEL2} C. I. Byrnes,  P. Enqvist,  and A. Lindquist, {\em Identifiability and 
well-posedness of shaping-filter parameterizations: A global analysis approach},
SIAM J. Control and Optimization, {\bf 41} (2002), 23--59.

\bibitem{BGL1} C. Byrnes, T.T. Georgiou, and A. Lindquist,
{\em A generalized entropy criterion for Nevanlinna-Pick interpolation:
A convex optimization approach to certain problems in systems and
control,} {\em IEEE Trans.\ on Automatic Control}, {\bf 45} (2001), 822-839.

\bibitem{BLmoments}
C. I. Byrnes and A. Lindquist,
The generalized moment problem with complexity constraint, {\em Integral Equations and Operator Theory} {\bf 56} (2006) 163--180.

 
\bibitem{BLkrein}
C. I. Byrnes and A. Lindquist, The moment problem for rational measures: convexity in the spirit of Krein,  in {\em Modern Analysis and Application: Mark Krein Centenary Conference},  Vol. I:  Operator Theory and Related Topics, Book Series: Operator Theory Advances and Applications Volume 190,  Birkh{\"a}user, 2009, pp. 157 -- 169.
\bibitem{Carli-FPP}
F. P. Carli and A. Ferrante and M. Pavon and G. Picci, A Maximum Entropy Solution of the Covariance Extension Problem for Reciprocal Processes, {\em IEEE Trans. Automatic Control} {\bf AC-56} (2011), 1999-2012.

\bibitem{CarliGeorgiou}
F. Carli, T.T. Georgiou, On the Covariance Completion Problem under a Circulant Structure, {\em IEEE Transactions on Automatic Control} {\bf 56}(4) (2011), pp. 918 -922.

\bibitem{Chiuso-F-P-05}
A. Chiuso and A. Ferrante and G. Picci, Reciprocal realization and modeling of textured images, {\em Proceedings of the 44rd IEEE Conference on Decision and Control}, 2005. 

\bibitem{Davis-79}
P. Davis, {\em Circulant Matrices}, John Wiley \& Sons, 1979.

\bibitem{Dempster-72}
Dempster, A. P.: Covariance selection, {\em Biometrics} {\bf 28}(1), pp. 157-175 (1972)


\bibitem{Du} J.~C.~Dunn, A projected Newton method for minimization problems with nonlinear inequality constraints, \emph{Numer. Math.}, 53: 377--409, 1988.


\bibitem{En} P.~Enqvist, A homotopy approach to rational covariance extension with degree constraint, \emph{Int. J. Appl. Math. Comput. Sci.}, 11(5): 1173--1201, 2001.

\bibitem{PEthesis}
P. Enqvist, {\em Spectral estimation by Geometric, Topological and Optimization Methods}, 
PhD thesis, Optimization and Systems Theory, KTH, Stockholm,
Sweden, 2001. 

\bibitem{PE}
P. Enqvist, A convex optimization approach to ARMA(n,m) model design from covariance and cepstrum data,
{\em SIAM Journal on Control and Optimization}, {\bf 43(3)}: 1011-1036, 2004.

 
\bibitem{Gthesis} 
T.T. Georgiou, {\em Partial Realization of Covariance Sequences}, Ph.D. thesis, CMST, University of Florida, Gainesville 1983.

\bibitem{Georgiou1} T.T. Georgiou, Realization of power spectra from partial covariances, {\em IEEE Trans. on Acoustics, Speech and Signal Processing} {\bf ASSP-35} (1987) 438-449.

\bibitem{Georgiou3} T.T. Georgiou, Solution of the general moment problem via a one-parameter imbedding,  {\em IEEE Trans. Aut. Contr.} {\bf AC-50} (2005) 811-826.

\bibitem{Georgiou-L-03}
T. T. Georgiou and A. Lindquist, Kullback-Leibler approximation of spectral density functions, {\em IEEE Trans. Information Theory} {\bf 49} (2003), 2910--2917. 


\bibitem{DM} C.~J.~Demeure, and C.~T.~Mullis, The Euclid algorithm and the fast computation of cross-covariance and autocovariance sequences, \emph{IEEE Trans. Acoust., Speech, Signal Processing}, 37: 545--552, 1989.
	
 
\bibitem{Kalman} R. E. Kalman, Realization of Covariance Sequences, Proc. Toeplitz Memorial Conference, Tel Aviv, Israel, 1981.

\bibitem{KreinNudelman}
M.G. Krein and A.A. Nudelman,
{\em The Markov Moment Problem and Extremal Problems},  American Mathematical
Society, Providence, Rhode Island, 1977.


\bibitem{KFL91}
A. J. Krener, R. Frezza, and B. C. Levy, Gaussian reciprocal processes andself-adjoint
differential equations of second order, {\em Stochastics and Stochastics Reports}, vol. 34, pp.
29-56, 1991.


\bibitem{Levy-F-02}
B. C. Levy and A. Ferrante, Characterization of stationary discrete-time {G}aussian Reciprocal Processes over a finite interval, {\em SIAM J. Matrix Anal. Appl.} {\bf 24} (2002), 334-355. 

\bibitem{Levy-F-K-90}
B. C. Levy and R. Frezza and A.J. Krener, Modeling and Estimation of discrete-time {G}aussian Reciprocal Processes, {\em IEEE Trans. Automatic Control} {\bf AC-35} (1990), 1013-1023.

 
\bibitem{LMPcirculant}
A. Lindquist and C. Masiero  and G. Picci, On the Multivariate Circulant Rational Covariance Extension Problem, {\em Proc. IEEE Conf on Decision and Control} (2013), Florence, Italy, 7155--7161. 

\bibitem{LPcirculant} A.~Lindquist, and G.~Picci, The circulant rational covariance extension problem: the complete solution, {\em IEEE Trans. on automatic control}, 58(11): 2848--2861, 2013.
	
\bibitem{LP} A.~Lindquist, and G.~Picci, Modeling of stationary periodic time series by ARMA representations, ArXiv e-prints, 2015.

\bibitem{LPBook}  A. Lindquist and G. Picci, {\it Linear Stochastic Systems: A Geometric Approach to Modeling, Estimation and Identification}, Springer series in Contemporary Mathematics. Springer Verlag, 2015.


\bibitem{Chiarathesis}
C. Masiero,   {\it Multivariate moment problems with applications to spectral estimation and physical layer security in wireless communications}  PhD thesis, Department of Information Engineering, University of Padova (2014). 

 
\bibitem{Pavon-F-12}
M. Pavon and . Ferrante, {\em On the Geometry of Maximum Entropy
Problems}, provisionally accepted for publication in {\em SIAM REVIEW},
available in {\tt http://arxiv.org/abs/1112.5529}, 2012.

\bibitem{Picci-C-08}
G. Picci and F. Carli, Modelling and simulation of images by reciprocal processes, {\em Proc. Tenth International Conference on Computer Modeling and Simulation UKSIM 2008}, 513--518.


	
	
\bibitem{Pi-MTNS016} G.~Picci,   A new approach to circulant band extension
{\em Proc. of the 22nd International Symposium on
Mathematical Theory of Networks and Systems (MTNS)}, July 11-15, 2016. Minneapolis, MN, USA, pp 123-130.

\bibitem{Rauch-S-T-65}
 H. E. Rauch, C. T. Striebel and F. Tung, Maximum likelihood estimates of linear dynamic systems {\em AIAA Journal}, {\bf 3} pp. 1445-1450 (1965).


\bibitem{Rissanen-73}
\newblock J. Rissanen, Algorithms for Triangular Decomposition of Block-Hankel and Toeplitz matrices with applications to factoring positive matrix Polynomials,	 {\em Math. Comp.}, {\bf 27} pp. 147-154, (1973).
	
\bibitem{SS} R.~H.~Shumway, and D.~S.~Stoffer, \emph{Time Series Analysis and Its Applications}, Springer, New York, 2011.

\bibitem{Whittle-63}	P. Whittle, On the fitting of multivariate autoregressions and the approximate canonical factorization of a spectral density, {\em Biometrica}, {\bf 50}, pp. 129-134  (1963). 
	









\end{thebibliography}
\end{document}